\pdfoutput=1
\RequirePackage{ifpdf}
\ifpdf 
\documentclass[pdftex]{sigma}
\else
\documentclass{sigma}
\fi

\numberwithin{equation}{section}

\newtheorem{Theorem}{Theorem}[section]
\newtheorem*{Theorem*}{Theorem}
\newtheorem{Corollary}[Theorem]{Corollary}

\newtheorem{Proposition}[Theorem]{Proposition}
 { \theoremstyle{definition}

\newtheorem{Example}[Theorem]{Example}
 }

\begin{document}
\allowdisplaybreaks

\newcommand{\arXivNumber}{2504.08293}

\renewcommand{\PaperNumber}{080}

\FirstPageHeading

\ShortArticleName{Skew Pl\"ucker Relations}

\ArticleName{Skew Pl\"ucker Relations}

\Author{Kazuya AOKAGE~$^{\rm a}$, Eriko SHINKAWA~$^b$ and Hiro-Fumi YAMADA~$^{\rm c}$}

\AuthorNameForHeading{K.~Aokage, E.~Shinkawa and H.-F.~Yamada}

\Address{$^{\rm a)}$~Department of Mathematics, National Institute of Technology, Ariake College,\\
\hphantom{$^{\rm a)}$}~Fukuoka 836-8585, Japan}
\EmailD{\href{mailto:aokage@ariake-nct.ac.jp}{aokage@ariake-nct.ac.jp}}

\Address{$^{\rm b)}$~Mathematical Science Center for Co-creative Society, Tohoku University,\\
\hphantom{$^{\rm b)}$}~Sendai 980-8577, Japan}
\EmailD{\href{mailto:eriko.shinkawa.e8@tohoku.ac.jp}{eriko.shinkawa.e8@tohoku.ac.jp}}

\Address{$^{\rm c)}$~Department of Mathematics, Rikkyo University, Tokyo 171-8501, Japan}
\EmailD{\href{mailto:hfyamada@rikkyo.ac.jp}{hfyamada@rikkyo.ac.jp}}

\ArticleDates{Received April 14, 2025, in final form September 17, 2025; Published online September 30, 2025}

\Abstract{Schur functions satisfy the relative Pl\"ucker relations which describe the projective embedding of the flag varieties and the Hirota bilinear equations for the modified KP hierarchies. These relative Pl\"ucker relations are generalized to the skew Schur functions.}


\Classification{05E05; 15A63; 37K10}

\section{Introduction}
Fix non-negative integer $M$, $n$ and $m$ satisfying $n+m\leq M$. Let
$Z={(z_{ij})}{}_{0\leq i<n+m, 0\leq j <M}$~be a~matrix of full rank, with entries in a field of characteristic $0$. Let us denote by \smash{$\xi^{k_0\dots k_{p-1}}_{\ell_{0}\dots \ell_{p-1}}$} the minor determinant consisting of rows $k_0,\dots,k_{p-1}$, and columns $\ell_0,\dots,\ell_{p-1}$. We write \smash{$\xi_{\ell_{0}\dots \ell_{p-1}}=\xi^{0\dots p-1}_{\ell_{0}\dots \ell_{p-1}}$}. These minor determinants satisfy the so-called relative Pl\"ucker relations
\begin{gather*}
\sum_{i=0}^{n+m}{(-1)}^{i}\xi_{k_{0}\dots k_{n-2}\ell_{i}}\xi_{\ell_{0}\dots \widehat{\ell_{i}}\,\dots \ell_{n+m}}=0
\end{gather*}
for arbitrary sequences $(k_0,\dots,k_{n-2})$ and $(\ell_0,\dots,\ell_{n+m})$ of length $n-1$ and $n+m+1$, respectively. These quadratic relations are defining equations of the flag variety. Details of this projective variety is found in the book \cite{gh}.

Since Schur functions have a determinant expression, the Jacobi--Trudi formula, they satisfy the relative Pl\"ucker relations. To be more precise, let $u=(u_1,u_2,\dots)$ be variables. For an~indeterminate $x$, put \smash{$\eta(u,x)=\sum_{j=1}^{\infty}u_jx^j$}, and define polynomials $h_n(u)$ by
\begin{gather*}
{\rm e}^{\eta(u,x)}=\sum_{}h_{n}(u)x^n.
\end{gather*}
The Schur function $S_{\ell_0 \dots \ell_{n-1}}(u)$ is expressed in terms of $h_{n}(u)$ as follows:
\begin{gather*}
S_{\ell_0 \dots \ell_{n-1}}(u)=\det{(h_{\ell_{j}-i}(u))}_{0\leq i,j<n}.
\end{gather*}
If $Z=(h_{j-i}(u))_{0\leq i <n+m, 0\leq j < M}$, then $S_{\ell_0 \dots \ell_{n-1}}(u)=S_{\ell_0 \dots \ell_{n-1}}$.
Here we explain briefly an~apparent connection between the relative Pl\"ucker relations and the Hirota bilinear equations of the modified KP hierarchies (see, for example, \cite{ostt}). The addition formula for the $\tau$-functions of the~modified KP hierarchy yields the following Hirota bilinear equations:
\begin{gather*}
\sum_{i=0}^{n+m}{(-1)}^{i}S_{k_0 k_1 \dots k_{n-2} \ell_i}\biggl(\frac{1}{2}\widetilde{D}\biggr)\cdot S_{\ell_0 \ell_1 \dots \widehat{\ell_i} \dots \ell_{n+m}}\biggl(-\frac{1}{2}\widetilde{D}\biggr)\tau^{[0]}\bullet \tau^{[m]}=0,
\end{gather*}
where $\widetilde{D} = \bigl(D_1, \frac{1}{2}D_2, \dots\bigr)$, $D_j = D_{t_j}$ being Hirota bilinear operator.
Here are two examples. In~the case $m=0$, $n=2$, and $(k_0) = (0)$, $(\ell_0,\ell_1, \ell_2)=(1,2,3)$, the Pl\"ucker relation reads
\begin{gather*}
\xi_{01}\xi_{23} - \xi_{02}\xi_{13} + \xi_{03}\xi_{12} = 0.
\end{gather*}
This is the defining equation of the $4$-dimensional Grassmann variety ${\rm GM}(2,4)$ in the projective space $P^5$.
The corresponding Hirota bilinear equation is
\begin{gather*}
\bigl(D_1^4 - 4D_1D_3 + 3D_2^2\bigr)\tau\bullet \tau =0.
\end{gather*}
The original KP equation
\begin{gather*}
3u_{yy} + (-4u_t + u_{xxx} -6uu_x)_x =0
\end{gather*}
is derived by the change of variables $x=t_1$, $y=t_2$, $t=t_3$ and $u=(\log \tau)_{xx}$.
In the case $m=1$, $n=2$, $(k) = (0)$ and $(\ell_0,\ell_1, \ell_2, \ell_3) = (0,1,2,3)$, the relative Pl\"ucker relation reads
\begin{gather*}
\xi_{01}\xi_{023} - \xi_{02}\xi_{013} + \xi_{03}\xi_{012} = 0.
\end{gather*}
The corresponding Hirota bilinear equation is
\begin{gather*}
\bigl(D_1^2 + D_2\bigr) \tau^{[0]}\bullet \tau^{[1]} = 0.
\end{gather*}
This is the so-called Miura transformation which connects the modified KP equation with the~KP equation.
We remark here that the Hirota bilinear equations for the (modified) KP and (modified) KdV hierarchies arise also in the Fock representations of the Virasoro algebra \cite{asy1,asy2}.

In our recent work \cite{asy3}, we proved a ``differential'' version of the Pl\"ucker relations for Schur functions and Schur Q-functions.
This is motivated by the above mentioned connection with the KP hierarchy. Our proof of differential Pl\"ucker relations is achieved by realizing differential of the Schur functions as skew Schur functions. In the present paper, we prove the relative version of the skew Pl\"ucker relations.

In Section~\ref{section2}, we recall the differential Pl\"ucker relations which appeared in \cite{asy3}.
Section~\ref{section3} is the main part of this note.
Let \smash{$H^{a-1}_{b}$} denote the hook Young diagram \smash{$\bigl(a, 1^b\bigr)$}. We will prove, for an even $m\geq 0$,
\begin{gather*}
\sum_{i=0}^{n+m}{(-1)}^{i}\sum_{(H^{a-1}_{b},H^{c-1}_{d})\in X}S_{k_0k_1\dots k_{n-2},\ell_i/H^{a-1}_{b}}(u)\cdot S_{\ell_0\ell_1\dots\widehat{\ell_i}\dots\ell_{n+m}/H^{c-1}_{d}}(u)=0.
\end{gather*}
Here \smash{$S_{\lambda/\mu}$} denotes the skew Schur function and $X$ denotes the set of the hook partitions of size~$N_{\geq0}$. Various differential Pl\"ucker relations can be derived from these skew Pl\"ucker relations.

Quadratic relations for the (skew) Schur functions are discussed, for example, in \cite{bbt, g, hb}, and references therein.
The formulas proved in the present paper are, apparently, not directly derived from these previous works.
Concrete correspondences with integrable systems, $\tau$-functions and Hirota bilinear equations are to be revealed.

\section{Relative Pl\"ucker relations}\label{section2}
The Schur functions are labeled by the sequences of non-negative integers $L=(\ell_0,\ell_1,\dots,\ell_{n-1})$ of length $n$. Schur functions are alternating in permutations of indices. If $L=(\ell_0,\ell_1,\dots,\ell_{n-1})$ is such that $0\leq\ell_{0}<\ell_{1} <\dots<\ell_{n-1}$, then, putting $\lambda_{i}=\ell_{n-i}-(n-i)$, $i = 1,\dots,n,$ one has a Young diagram $\lambda=(\lambda_1,\lambda_2,\dots,\lambda_{n})$ of length at most $n$. The strictly increasing sequence $(\ell_0,\ell_1,\dots,\ell_{n-1})$ defines the same Young diagram as the sequence $(0,1,\dots,k-1,\allowbreak{\ell_{0}+k}, {\ell_{1}+k},\dots,\ell_{n-1}+k)$. Hence, the Schur functions labeled by these two sequences coincide. For example,
\begin{gather*}
S_{1,3}(u)=S_{0,2,4}(u)=-S_{4,2,0}(u)=-S_{3,1}(u)=-S_{0,1,5,3}(u).
\end{gather*}
Let $K=(k_0,k_1,\dots,k_{n-2})$ and $L=(\ell_{0},\ell_{1},\dots,\ell_{n+m})$ be sequences of length
$n-1$ and $n+m+1$, respectively. By the $m$-relative Pl\"ucker relation, we mean the following:
\begin{gather}
\sum_{i=0}^{n+m}{(-1)}^{i}S_{k_0k_1\dots k_{n-2},\ell_i}(u)\cdot S_{\ell_0\ell_1\dots\widehat{\ell_i}\dots\ell_{n+m}}(u)=0, \label{modi-plu}
\end{gather}
where the symbol \,$\widehat{\ }$\, denotes the deletion. The equation (\ref{modi-plu}) is derived from Laplace expansion of the
$(2n+m) \times (2n+m)$ determinant
\begin{align*}
\left|\begin{array}{cccc|ccc}
h_{k_0} & h_{k_{0}-1} & \hdots & h_{k_{0}-(n-1)} & 0 & \hdots & 0 \\
h_{k_1} & h_{k_{1}-1} & \hdots & h_{k_{1}-(n-1)} & 0 & \hdots & 0 \\
\vdots & \vdots & \ddots & \vdots & \vdots & \ddots & \vdots \\
h_{k_{n-2}} & h_{k_{n-2}-1} & \hdots & h_{k_{n-2}-(n-1)} & 0 &\hdots & 0 \\
h_{\ell_0} & h_{\ell_{0}-1} & \hdots & h_{\ell_{0}-(n-1)} & h_{\ell_{0}} & \hdots & h_{\ell_{0}-(n+m-1)} \\[1mm]\hline
h_{\ell_1} & h_{\ell_{1}-1} & \hdots & h_{\ell_{1}-(n-1)} & h_{\ell_{1}} & \hdots & h_{\ell_{1}-(n+m-1)} \\
\vdots & \vdots & \ddots & \vdots & \vdots & \ddots & \vdots \\
h_{\ell_{n+m}} & h_{\ell_{n+m}-1} & \hdots & h_{\ell_{n+m}-(n-1)} & h_{\ell_{n+m}} & \hdots & h_{\ell_{n+m}-(n+m-1)}
\end{array}\right|
=0.
\end{align*}

For example, if we take $K=(0,1)$, $L=(1,2,3,4,5,6)$, then we have $1$-relative Pl\"ucker relation
\begin{gather*}
S_{0,1,2}(u)\cdot S_{3,4,5,6}(u)-S_{0,1,3}(u)\cdot S_{2,4,5,6}(u)+S_{0,1,4}(u)\cdot S_{2,3,5,6}(u)\\
\qquad{}-S_{0,1,5}(u)\cdot S_{2,3,4,6}(u)+S_{0,1,6}(u)\cdot S_{2,3,4,5}(u)=0.
\end{gather*}

For sequences $L=(\ell_0,\ell_2,\dots,\ell_{n-1})$ and $R=(r_0,r_1,\dots,r_{n-1})$, we define the skew Schur function
\begin{gather*}
S_{L/R}(u)=\det{(h_{\ell_i-r_{j}})}_{0\leq i,j< n}.
\end{gather*}
For example, if we take $L=(\ell_0,\ell_1,\ell_{2},\ell_{3})$ and $R=(r_0,r_1)$, then we think $R=(0,1,r_{0}+2,r_{1}+2)$,
\begin{gather*}
S_{L/R}(u)=
\begin{vmatrix}
h_{\ell_0} & h_{\ell_0-1} & h_{\ell_0-(r_0+2)} & h_{\ell_0-(r_1+2)} \\
h_{\ell_1} & h_{\ell_1-1} & h_{\ell_1-(r_0+2)} & h_{\ell_1-(r_1+2)} \\
h_{\ell_2} & h_{\ell_2-1} & h_{\ell_2-(r_0+2)} &h_{\ell_2-(r_1+2)}\\
h_{\ell_3} & h_{\ell_3-1} & h_{\ell_3-(r_0+2)} & h_{\ell_3-(r_1+2)}
\end{vmatrix}.
\end{gather*}
With respect to the Hall inner product $\langle\ ,\, \rangle$ (cf.\ \cite{mac}), we have
\[
\langle \partial_{u_n} S_{L}(u), S_{K}(u)\rangle =\langle S_{L}(u), p_nS_{K}(u)\rangle,\nonumber
\]
where $p_n$ is the power sum symmetric function of degree $n$.
The following is known as the Frobenius formula \cite{mac}
\begin{gather*}
p_{n}=\sum_{R\in H(n)}{(-1)}^{{\rm leg}(R)}S_{R}(u).\nonumber
\end{gather*}
Here, $H(n)$ denotes the set of sequences that correspond to hook partitions of size $n$, and ${\rm leg}(R)$ denotes the leg length of the corresponding hook partition. We see
\begin{align*}
\langle \partial_{u_n} S_{L}(u), S_{K}(u)\rangle &=\sum_{R\in H(n)}{(-1)}^{{\rm leg}(R)}\langle S_{L}(u),S_{R}(u)S_{K}(u)\rangle\nonumber\\
&=\sum_{R\in H(n)}{(-1)}^{{\rm leg}(R)}\langle S_{L/R}(u),S_{K}(u)\rangle.\nonumber
\end{align*}
Also, we define
\begin{gather*}
S_L\bigl(\tilde{\partial}\bigr)=S_{L}(u)|_{u_j=\frac{1}{j}\partial_{u_j}},
\end{gather*}
where $\tilde{\partial}=\bigl(\partial_{u_1},\frac{1}{2}\partial_{u_2},\frac{1}{3}\partial_{u_3},\dots\bigr)$. We see
\begin{align*}
\bigl\langle S_{R}\bigl(\tilde{\partial}\bigr)S_{L}(u), S_{K}(u)\bigr\rangle =\langle S_{L}(u),S_{R}(u)S_{K}(u)\rangle
=\langle S_{L/R}(u), S_{K}(u)\rangle.
\end{align*}
Hence
\begin{gather*}
\partial_{u_n}S_{L}(u)=\sum_{R\in H(n)}{(-1)}^{{\rm leg}(R)}S_{L/R}(u)\qquad \mbox{and}\qquad S_{R}\bigl(\tilde{\partial}\bigr)S_{L}(u)=S_{L/R}(u).
\end{gather*}
The following differential version of the Pl\"ucker relations are shown in our previous paper \cite{asy3}, by making use of the skew Schur functions.

\begin{Theorem}
Suppose that $N\geq0$. For sequences $(k_{0},k_1,\dots,k_{n-2})$ and $(\ell_0,\ell_1,\dots, \ell_{n})$, we~have
\begin{gather*}
\sum_{i= 0}^{n}{(-1)}^{i}\Biggl(\sum_{\substack{\alpha+\beta=N\\ \alpha\geq 0,\,\beta \geq 0}}\partial_{u_\alpha}S_{k_0 k_1 \dots k_{n-2},\ell_i}(u)\cdot \partial_{u_\beta}S_{\ell_0 \ell_1 \dots \widehat{\ell_i} \dots \ell_{n}}(u) \Biggr)=0.
\end{gather*}
\end{Theorem}

\section{Skew Pl\"ucker relations}\label{section3}
Denote by $H^{a-1}_{b}$ the hook partition $\bigl(a,1^b\bigr)$ with arm length $a-1$ and leg length $b$.
We write $0$ in place of $H^{-1}_{0}$. For example, $H(0)=\{0\}$, $H(1)=\bigl\{H^{0}_{0}\bigr\}$, $H(2)=\bigl\{H^{1}_{0},H^{0}_{1}\bigr\}$. For the sequences $(k_{0},k_1,\dots,k_{n-2})$ and $(\ell_0,\ell_1,\dots, \ell_{n+m})$, we have
\begin{align}
&\sum_{i=0}^{n+m}{(-1)}^{i}S_{k_0 k_1\dots k_{n-2}\ell_i/{H^{a-1}_{b}}}(u)\cdot S_{\ell_0 \ell_1\dots \widehat{\ell_{i}}\dots\ell_{n+m}/{H^{c-1}_{d}}}(u)\label{skewplu}\\
&=\left|\begin{array}{cccc|cccc}
h_{k_0-r_{0}^{}} & h_{k_0-r_{1}^{}} & \cdots &h_{k_0-r_{n-1}^{}}& 0 & 0 & \cdots & 0\\
h_{k_{1}-r_{0}^{}} & h_{k_{1}-r_{1}^{}} & \cdots &h_{k_{1}-r_{n-1}^{}}& 0 & 0 & \cdots & 0 \\
\vdots & \vdots & \ddots & \vdots & \vdots &\vdots & \ddots & \vdots \\
h_{k_{n-2}-r_{0}^{}} & h_{k_{n-2}-r_{1}^{}} & \cdots &h_{k_{n-2}-r_{n-1}^{}} &0 & 0 & \cdots & 0 \\
h_{\ell_0-r_{0}^{}} & h_{\ell_0-r_{1}^{}} & \cdots &h_{{\ell_0}-r_{n-1}^{}}& h_{\ell_{0}-w_{0}^{}} & h_{\ell_{0}-w_{1}^{}} & \cdots & h_{\ell_{0}-w_{n+m-1}^{}} \\[1mm]\hline
h_{\ell_1-r_{0}^{}} & h_{\ell_1-r_{1}^{}} & \cdots &h_{{\ell_1}-r_{n-1}^{}}& h_{\ell_{1}-w_{0}^{}}& h_{\ell_{1}-w_{1}^{}} & \cdots & h_{\ell_{1}-w_{n+m-1}^{}}\\
\vdots & \vdots& \ddots & \vdots &\vdots &\vdots & \ddots & \vdots \\
h_{\ell_{n+m}-r_{0}^{}} & h_{\ell_{n+m}-r_{1}^{}} & \cdots &h_{{\ell_{n+m}}-r_{n-1}^{}}& h_{\ell_{n+m}-w_{0}^{}} & h_{\ell_{n+m}-w_{1}^{}} & \cdots & h_{\ell_{n+m}-w_{n+m-1}^{}}
\end{array}\right|\!,\nonumber
\end{align}
where
\begin{gather*}
(r_0,r_1,\dots,r_{n-(b+2)},r_{n-(b+1)},\dots,r_{n-2},r_{n-1})\\
\qquad{}=(0,1,\dots,n-(b+2),n-(b+1)+1,\dots,(n-2)+1,(n-1)+a),
\end{gather*}
and
\begin{gather*}
(w_0,w_1,\dots,w_{n+m-(d+2)},w_{n+m-(d+1)},\dots,w_{n+m-2},w_{n+m-1})\\
=(0,1,\dots,n+m-(b+2),n+m-(b+1)+1,\dots,(n+m-2)+1,(n+m-1)+c).
\end{gather*}

Denote the right-hand side of the equation (\ref{skewplu}) by $\xi\bigl(H^{a-1}_{b},H^{c-1}_{d}\bigr)$ and denote by
\[
\xi\bigl(H^{a-1}_{b},H^{c-1}_{d}\bigr)^{i_0 i_1 \dots i_{k-1}}_{j_0 j_1 \dots j_{k-1}}
\]
the $k\times k$ minor determinant of $\xi\bigl(H^{a-1}_{b},H^{c-1}_{d}\bigr)$ consisting of $i_0$-th row, $i_1$-th row, $\dots$ and $j_0$-th column, $j_1$-th column, $\dots$.

By performing column operations on \smash{$\xi\bigl(H^{a-1}_{b},H^{c-1}_{d}\bigr)$}, all entries below the $(n-1)$-th row become zero in some columns.
In practice, for $m\geq0$, column operations on $\xi\bigl(H^{a-1}_{b},H^{c-1}_{d}\bigr)$ lead to at most two non-zero terms
$h_{\ell_{0}-(n+m)+d+1}$ and $h_{\ell_{0}-n+1-a}$ appearing in the $n$-th row. For convenience, we write the determinant \smash{$\xi\bigl(H^{a-1}_{b},H^{c-1}_{d}\bigr)$} only by the indices of~$h$ in the $(n-1)$-th row, though the components in the $(n-2)$-th row should always be considered.
The indices of~$h$ in the $(n-1)$-th row are, from the left,
\begin{align}
&| \underbrace{\ell_{0},\ell_{0}-1,\dots,\ell_{0}-(n-(b+2))}_{n-b-1},\nonumber\\
&\quad{}\underbrace{\ell_{0}-(n-(b+1))-1,\dots,\ell_{0}-(n-2)-1}_{b},\ell_{0}-(n-1)-a,\nonumber\\
&\quad{}\underbrace{\ell_{0},\ell_{0}-1,\dots,\ell_{0}-((n+m)-(d+2))}_{n+m-d-1},\nonumber\\
&\quad{}\underbrace{\ell_{0}-((n+m)-(d+1))-1,\dots,\ell_{0}-((n+m)-2)-1}_{d},\ell_{0}-((n+m)-1)-c |.\label{ind}
\end{align}
In addition, if a certain entry in the $(n-1)$-th row is zero, it is written as $0$ instead of the index of $h$. After performing column operations, it is represented using the indices of $h$ in the $(n-1)$-th row. For example, when $n=3$, $m=1$, $K=(k_0,k_1)$, $L=(\ell_{0},\ell_{1},\ell_{2},\ell_{3},\ell_4)$, $a=2$, $b=2$, $c=2$, $d=1$, we have
\begin{align*}
\xi\bigl(H^{1}_{2},H^{1}_{1}\bigr)&=\sum_{i= 0}^{4}{(-1)}^{i}S_{k_0 k_1 \ell_i/H^{1}_{2}}(u)\cdot S_{\ell_0 \dots \widehat{\ell_i} \dots \ell_{4}/H^{1}_{1}}(u)\nonumber\\
&=\left|\begin{array}{ccc|cccc}
h_{k_0-1} & h_{k_0-1-1} &h_{k_0-2-2} & 0 & 0 & 0 &0\\
h_{k_1-1} & h_{k_1-1-1} &h_{k_1-2-2} & 0 & 0 & 0 &0\\
h_{\ell_0-1} & h_{\ell_0-1-1} &h_{\ell_0-2-2} & h_{\ell_{0}} &h_{\ell_{0}-1} & h_{\ell_{0}-2-1} & h_{\ell_{0}-3-2} \\\hline
h_{\ell_1-1} & h_{\ell_1-1-1} &h_{\ell_1-2-2} & h_{\ell_{1}} &h_{\ell_{1}-1}& h_{\ell_{1}-2-1} & h_{\ell_{1}-3-2}\\
h_{\ell_2-1} & h_{\ell_2-1-1} &h_{\ell_2-2-2} & h_{\ell_{2}} &h_{\ell_{2}-1}& h_{\ell_{2}-2-1} & h_{\ell_{2}-3-2}\\
h_{\ell_3-1} & h_{\ell_3-1-1} &h_{\ell_3-2-2} & h_{\ell_{3}} &h_{\ell_{3}-1}& h_{\ell_{3}-2-1} & h_{\ell_{3}-3-2}\\
h_{\ell_4-1} & h_{\ell_4-1-1} &h_{\ell_4-2-2} & h_{\ell_{4}} &h_{\ell_{4}-1}& h_{\ell_{4}-2-1} & h_{\ell_{4}-3-2}
\end{array}\right|\nonumber\\
&=\left|\begin{array}{ccc|cccc}
h_{k_0-1} & h_{k_0-1-1} &h_{k_0-2-2} & 0 & 0 & 0 &0\\
h_{k_1-1} & h_{k_1-1-1} &h_{k_1-2-2} & 0 & 0 & 0 &0\\
0 & h_{\ell_0-1-1} & h_{k_1-2-2} & h_{\ell_{0}} &h_{\ell_{0}-1} & h_{\ell_{0}-2-1} & h_{\ell_{0}-3-2} \\\hline
0 & h_{\ell_1-1-1} & h_{k_1-2-2} & h_{\ell_{1}} &h_{\ell_{1}-1}& h_{\ell_{1}-2-1} & h_{\ell_{1}-3-2}\\
0 & h_{\ell_2-1-1} & h_{k_1-2-2} & h_{\ell_{2}} &h_{\ell_{2}-1}& h_{\ell_{2}-2-1} & h_{\ell_{2}-3-2}\\
0 & h_{\ell_3-1-1} & h_{k_1-2-2} & h_{\ell_{3}} &h_{\ell_{3}-1}& h_{\ell_{3}-2-1} & h_{\ell_{3}-3-2}\\
0 & h_{\ell_4-1-1} & h_{k_1-2-2} & h_{\ell_{4}} &h_{\ell_{4}-1}& h_{\ell_{4}-2-1} & h_{\ell_{4}-3-2}
\end{array}\right|.\nonumber
\end{align*}
Hence, we have
\begin{align*}
\xi\bigl(H^{1}_{2},H^{1}_{1}\bigr)&=|\ell_{0}-1,\ell_{0}-2,\ell_{0}-4,\ell_{0},\ell_{0}-1,\ell_{0}-3,\ell_{0}-5|\\
&=|0,{\ell_{0}-2},{\ell_0-4},\ell_{0},{\ell_{0}-1},\ell_{0}-3,\ell_{0}-5|.
\end{align*}

It is essential to consider the following points. For example, \smash{$\xi\bigl(H^{2}_{2},H^{0}_{1}\bigr)$} can be rewritten as $ |\ell_{0}-1,\ell_{0}-2,\ell_{0}-5,\ell_{0},\ell_{0}-1,\ell_{0}-3,\ell_{0}-4|=|0,\ell_{0}-2,\ell_0-5,\ell_{0},\ell_{0}-1,\ell_{0}-3,\ell_{0}-4|$. However, \smash{$\xi\bigl(H^{1}_{2},H^{1}_{1}\bigr)$} and \smash{$\xi\bigl(H^{2}_{2},H^{0}_{1}\bigr)$} do not coincide, even when scaled by $\pm 1$. Note that this phenomenon occurs only when exactly two non-zero terms $h_{\ell_{0}-(n+m)+d+1}$ and $h_{\ell_{0}-n+1-a}$, appear in the $n$-th~row. 

In \cite{asy3}, we have presented all the linear relations of $\xi\bigl(H^{a-1}_{b},H^{c-1}_{d}\bigr)$ for $m=0$. Now we examine the case $m\geq1$. From (\ref{ind}), we have
\begin{gather}
\xi\bigl(H^{c-1+m}_{b},H^{a-1-m}_{d}\bigr)\nonumber\\
\qquad=|\underbrace{\ell_{0},\ell_{0}-1,\dots,\ell_{0}-n+b+2}_{n-b-1},
\underbrace{\ell_{0}-n+b,\dots,\ell_{0}-n+1}_{b},\ell_{0}-n-m+1-c,\nonumber\\
\qquad\hphantom{=|}{} \underbrace{\ell_{0},\ell_{0}-1,\dots,\ell_{0}-(n+m)+d+2}_{n+m-d-1},\nonumber\\
\qquad\hphantom{=|}{} \underbrace{\ell_{0}-(n+m)+d,\dots,\ell_{0}-(n+m)+1}_{d},\ell_{0}-n+1-a|,\label{ind1}
\\
\xi\bigl(H^{a-1}_{d-m},H^{c-1}_{b+m}\bigr)\nonumber\\
\qquad=|\underbrace{\ell_{0},\ell_{0}-1,\dots,\ell_{0}-n+d-m+2}_{n+m-d-1},
\underbrace{\ell_{0}-n+d-m,\dots,\ell_{0}-n+1}_{-m+d},\ell_{0}-n+1-a,\nonumber\\
\hphantom{\qquad=|}{} \underbrace{\ell_{0},\ell_{0}-1,\dots,\ell_{0}-n+b+2}_{n-b-1},\nonumber\\
\hphantom{\qquad=|}{} \underbrace{\ell_{0}-n+b,\dots,\ell_{0}-(n+m)+1}_{m+b},\ell_{0}-(n+m)+1-c|\label{ind2}
\end{gather}
and
\begin{gather}
\xi\bigl(H^{c-1+m}_{d-m},H^{a-1-m}_{b+m}\bigr)\nonumber\\
\qquad=|\underbrace{\ell_{0},\ell_{0}-1,\dots,\ell_{0}-n-m+d+2}_{n+m-d-1},
\underbrace{\ell_{0}-n-m+d,\dots,\ell_{0}-n+1}_{-m+d},\nonumber\\
\qquad\hphantom{=|}{}\, \ell_{0}-n-m+1-c,\underbrace{\ell_{0},\ell_{0}-1,\dots,\ell_{0}-n+b+2}_{n-b-1},\nonumber\\
\qquad\hphantom{=|}{} \underbrace{\ell_{0}-n+b,\dots,\ell_{0}-(n+m)+1}_{m+b},\ell_{0}-n+1-a|.\label{ind3}
\end{gather}
\begin{Proposition}\label{00}\samepage
Suppose that $m \geq 1$.
\begin{enumerate}
\item[$(1)$] For $1\leq a\leq m$, we have\
\begin{gather*}
\sum_{i=0}^{n+m}{(-1)}^{i}S_{k_0 k_1\dots k_{n-2}\ell_i/{H^{a-1}_{b}}}(u)\cdot S_{\ell_0 \ell_1\dots \widehat{\ell_{i}}\dots\ell_{n+m}/{(0)}}(u)=0,
\end{gather*}
\item[$(2)$] For $c\geq 1$, $0\leq d \leq m-1$, we have\
\begin{gather*}
\sum_{i=0}^{n+m}{(-1)}^{i}S_{k_0 k_1\dots k_{n-2}\ell_i/{(0)}}(u)\cdot S_{\ell_0 \ell_1\dots \widehat{\ell_{i}}\dots\ell_{n+m}/{H^{c-1}_{d}}}(u)=0.
\end{gather*}
\end{enumerate}
\end{Proposition}
\begin{proof}
(1) By (\ref{ind}), we have
\begin{gather*}
\xi\bigl(H^{a-1}_{b},0\bigr)=|\underbrace{\ell_{0},\ell_{0}-1,\dots,\ell_{0}-n+b+2}_{n-b-1},\underbrace{\ell_{0}-n+b,\dots,\ell_{0}-n+1}_{b},\ell_{0}-n+1-a,\nonumber\\
\hphantom{\xi\bigl(H^{a-1}_{b},0\bigr)=|}{}
\underbrace{\ell_{0},\ell_{0}-1,\dots,\ell_{0}-n-m+2}_{n+m-1},\ell_{0}-n-m+1|.\nonumber
\end{gather*}
From the condition $1\leq a\leq m$, we have $\ell_0-n+1-a\geq\ell_{0}-n+1-m=\ell_{0}-n-m+1$. Hence
\begin{gather*}
\xi\bigl(H^{a-1}_{b},0\bigr)=|\underbrace{0,0,\dots,0}_{n},\underbrace{\ell_{0},\ell_{0}-1,\dots,\ell_{0}-n-m+2}_{n+m-1},\ell_{0}-n-m+1|=0.\nonumber
\end{gather*}

(2) By (\ref{ind}), we have
\begin{gather*}
\xi\bigl(0,H^{c-1}_{d}\bigr)=|\underbrace{\ell_{0},\ell_{0}-1,\dots,\ell_{0}-n+2}_{n-1},\ell_{0}-n+1,
\underbrace{\ell_{0},\ell_{0}-1,\dots,\ell_{0}-n-m+d+2}_{n+m-d-1},\nonumber\\
\hphantom{\xi\bigl(0,H^{c-1}_{d}\bigr)=|}{}
\underbrace{\ell_{0}-n-m+d,\dots,\ell_{0}-n-m+1}_{d},\ell_{0}-n-m+1-c|\nonumber.
\end{gather*}
From the condition $0\leq d \leq m-1$, we have $\ell_0-n-m+d+2\leq\ell_{0}-n-m+m-1+2=\ell_{0}-n+1$. Hence
\begin{gather*}
\xi\bigl(0,H^{c-1}_{d}\bigr)=|\underbrace{0,0,\dots,0}_{n},\underbrace{\ell_{0},\ell_{0}-1,\dots,\ell_{0}-n-m+1}_{n+m-1},\ell_{0}-n-m+1-c|=0.
\tag*{\qed}
\end{gather*}
\renewcommand{\qed}{}
\end{proof}
\begin{Proposition}\label{mabcd}
Suppose that $m\geq 1$, $a,c\geq 1$ and $a+d=m$. We have
\begin{gather}
\sum_{i=0}^{n+m}{(-1)}^{i}S_{k_0 k_1\dots k_{n-2}\ell_i/{H^{a-1}_{b}}}(u)\cdot S_{\ell_0 \ell_1\dots \widehat{\ell_{i}}\dots\ell_{n+m}/{H^{c-1}_{d}}}(u)\label{mabcd0}\\
\qquad{}={(-1)}^{j}\sum_{i=0}^{n+m}{(-1)}^{i}S_{k_0 k_1\dots k_{n-2}\ell_i/{H^{a-1-j}_{b}}}(u)\cdot S_{\ell_0 \ell_1\dots \widehat{\ell_{i}}\dots\ell_{n+m}/{H^{c-1}_{d+j}}}(u)\hspace{0.5cm}\ \ \ \ \ \ \ \ \label{mabcd1}\\
\qquad{}={(-1)}^{r}\sum_{i=0}^{n+m}{(-1)}^{i}S_{k_0 k_1\dots k_{n-2}\ell_i/{H^{a-1+r}_{b}}}(u)\cdot S_{\ell_0 \ell_1\dots \widehat{\ell_{i}}\dots\ell_{n+m}/{H^{c-1}_{d-r}}}(u) \label{mabcd2}\\
\qquad{}={(-1)}^{m-a-1}\sum_{i=0}^{n+m}{(-1)}^{i}S_{k_0 k_1\dots k_{n-2}\ell_i/{H^{c-1+m}_{b}}}(u)\cdot S_{\ell_0 \ell_1\dots \widehat{\ell_{i}}\dots\ell_{n+m}/{(0)}}(u)
\label{mabcd3}\\
\qquad{}={(-1)}^{a}\sum_{i=0}^{n+m}{(-1)}^{i}S_{k_0 k_1\dots k_{n-2}\ell_i/{(0)}}(u)\cdot S_{\ell_0 \ell_1\dots \widehat{\ell_{i}}\dots\ell_{n+m}/{H^{c-1}_{b+m}}}(u),
\label{mabcd4}
\end{gather}
where $1\leq j \leq a-1$ and $1\leq r\leq d$.
\end{Proposition}
\begin{proof}
From $a+b=m$, we have
\begin{gather}
\xi\bigl(H^{a-1}_{b},H^{c-1}_{d}\bigr)=|\underbrace{\ell_{0},\ell_{0}-1,\dots,\ell_{0}-n+b+2}_{n-b-1},\nonumber\\
\hphantom{\xi\bigl(H^{a-1}_{b},H^{c-1}_{d}\bigr)=|}{}
\underbrace{\ell_{0}-n+b,\dots,\ell_{0}-n+1}_{b},\ell_{0}-n+1-a,\nonumber\\
\hphantom{\xi\bigl(H^{a-1}_{b},H^{c-1}_{d}\bigr)=|}{}
\underbrace{\ell_{0},\ell_{0}-1,\dots,\ell_{0}-n-a+2}_{n+a-1},\nonumber\\
\hphantom{\xi\bigl(H^{a-1}_{b},H^{c-1}_{d}\bigr)=|}{}
\underbrace{\ell_{0}-n-a,\dots,\ell_{0}-(n+m)+1}_{d},\ell_{0}-(n+m)+1-c|.
\label{m0}
\end{gather}
By performing column operations, we have
\begin{gather*}
\xi\bigl(H^{a-1}_{b},H^{c-1}_{d}\bigr)=|\underbrace{0,\dots,0}_{n-b-1},
\underbrace{0,\dots,0}_{b},\ell_{0}-n+1-a,
\underbrace{\ell_{0},\ell_{0}-1,\dots,\ell_{0}-n-a+2}_{n+a-1},\nonumber\\
\hphantom{\xi\bigl(H^{a-1}_{b},H^{c-1}_{d}\bigr)=|}{}
\underbrace{\ell_{0}-n-a,\dots,\ell_{0}-(n+m)+1}_{d},\ell_{0}-(n+m)+1-c|\nonumber.
\end{gather*}
Let $X$ denote the above determinant with rearranged columns,
\begin{gather*}
| \underbrace{0,\dots,0}_{n-1},
\ell_{0},\ell_{0}-1,\dots,\ell_{0}-n-a+2,\\
\quad{}\ell_{0}-n-a+1,\dots,\ell_{0}-(n+m)+1,\ell_{0}-(n+m)+1-c|\nonumber.
\end{gather*}
Clearly, we have
\smash{$X={(-1)}^{n+a-1}\xi\bigl(H^{a-1}_{b},H^{c-1}_{d}\bigr)$}.

(i) We show that (\ref{mabcd0}) = (\ref{mabcd1}). By (\ref{m0}), we have
\begin{gather*}
\xi\bigl(H^{a-1-j}_{b},H^{c-1}_{d+j}\bigr)=|\underbrace{\ell_{0},\ell_{0}-1,\dots,\ell_{0}-n+b+2}_{n-b-1},\underbrace{\ell_{0}-n+b,\dots,\ell_{0}-n+1}_{b},\nonumber\\
\hphantom{\xi\bigl(H^{a-1-j}_{b},H^{c-1}_{d+j}\bigr)=|}{}\,
\ell_{0}-n+1-a+j,\underbrace{\ell_{0},\ell_{0}-1,\dots,\ell_{0}-n-a+2+j}_{n+a-1-j},\nonumber\\
\hphantom{\xi\bigl(H^{a-1-j}_{b},H^{c-1}_{d+j}\bigr)=|}{}
\underbrace{\ell_{0}-n-a+j,\dots,\ell_{0}-(n+m)+1}_{d+j},\ell_{0}-(n+m)+1-c|\nonumber.
\end{gather*}
From the inequality $\ell_0-n+1\geq\ell_0-n-a+2+j$, we have
\begin{gather*}
|\underbrace{0,\dots,0}_{n-1},\ell_{0}-n+1-a+j,\underbrace{\ell_{0},\ell_{0}-1,\dots,\ell_{0}-n-a+2+j}_{n+a-1-j},\nonumber\\
\quad{}\underbrace{\ell_{0}-n-a+j,\dots,\ell_{0}-(n+m)+1}_{d+j},\ell_{0}-(n+m)+1-c|\nonumber.
\end{gather*}
Therefore, \smash{$X={(-1)}^{n+a-1-j}\xi\bigl(H^{a-1-j}_{b},H^{c-1}_{d+j}\bigr)$}. Thus, \smash{$\xi\bigl(H^{a-1}_{b},H^{c-1}_{d}\bigr)={(-1)}^{j}\xi\bigl(H^{a-1-j}_{b},H^{c-1}_{d+j}\bigr)$}.

(ii) We show that (\ref{mabcd0}) = (\ref{mabcd2}). By (\ref{m0}), we have
\begin{gather*}
\xi\bigl(H^{a-1+r}_{b},H^{c-1}_{d-r}\bigr)=|\underbrace{\ell_{0},\ell_{0}-1,\dots,\ell_{0}-n+b+2}_{n-b-1},\underbrace{\ell_{0}-n+b,\dots,\ell_{0}-n+1}_{b},\\
\hphantom{\xi\bigl(H^{a-1+r}_{b},H^{c-1}_{d-r}\bigr)=|}{}
\,\ell_{0}-n+1-a-r,\underbrace{\ell_{0},\ell_{0}-1,\dots,\ell_{0}-n-a-r+2}_{n+a-1+r},\\
\hphantom{\xi\bigl(H^{a-1+r}_{b},H^{c-1}_{d-r}\bigr)=|}{}
\underbrace{\ell_{0}-n-a-r,\dots,\ell_{0}-(n+m)+1}_{d-r},\ell_{0}-(n+m)+1-c|.
\end{gather*}
From the inequality $\ell_0-n+1\geq\ell_0-n-a-r+2$, we have
\begin{gather*}
|\underbrace{0,\dots,0}_{n-1},\ell_{0}-n+1-a-r,\underbrace{\ell_{0},\ell_{0}-1,\dots,\ell_{0}-n-a-r+2}_{n+a-1+r},\nonumber\\
\quad{}\underbrace{\ell_{0}-n-a-r,\dots,\ell_{0}-(n+m)+1}_{d-r},\ell_{0}-(n+m)+1-c|\nonumber.
\end{gather*}
Therefore, \smash{$X={(-1)}^{n+a-1+r}\xi\bigl(H^{a-1+r}_{b},H^{c-1}_{d-r}\bigr)$}. Thus, \smash{$\xi\bigl(H^{a-1}_{b},H^{c-1}_{d}\bigr)={(-1)}^{r}\xi\bigl(H^{a-1+r}_{b},H^{c-1}_{d-r}\bigr)$}.

(iii) We show that (\ref{mabcd0})=(\ref{mabcd3}). By (\ref{m0}), we have
\begin{gather*}
\xi\bigl(H^{c-1+m}_{b},0\bigr)=|\underbrace{\ell_{0},\ell_{0}-1,\dots,\ell_{0}-n+b+2}_{n-b-1},
\underbrace{\ell_{0}-n+b,\dots,\ell_{0}-n+1}_{b},\nonumber\\
\hphantom{\xi\bigl(H^{c-1+m}_{b},0\bigr)=|}{}
\,\ell_{0}-(n+m)+1-c,\underbrace{\ell_{0},\ell_{0}-1,\dots,\ell_{0}-(n+m)+2}_{n+m-1},
\ell_{0}-(n+m)+1|\nonumber.
\end{gather*}
From the inequality $\ell_{0}-n+1\geq \ell_0-(n+m)+2$, we have
\begin{gather*}
|\underbrace{0,\dots,0}_{n-b-1},
\underbrace{0,\dots,0}_{b},\ell_{0}-(n+m)+1-c,\nonumber\\
\quad{}\underbrace{\ell_{0},\ell_{0}-1,\dots,\ell_{0}-(n+m)+2}_{n+m-1},
\underbrace{\ell_{0}-(n+m)+1}_{1}|\nonumber.
\end{gather*}
Therefore, \smash{$X={(-1)}^{n+m}\xi\bigl(H^{c-1+m}_{b},0\bigr)$}. Thus, \smash{$\xi\bigl(H^{a-1}_{b},H^{c-1}_{d}\bigr)={(-1)}^{a-1+m}\xi\bigl(H^{c-1+m}_{b},0\bigr)$}.

(iv) We show that (\ref{mabcd0}) = (\ref{mabcd4}). By (\ref{m0}), we have
\begin{gather*}
\xi\bigl(0,H^{c-1}_{b+m}\bigr)=|\underbrace{\ell_{0},\ell_{0}-1,\dots,\ell_{0}-n+2}_{n-1}, {\ell_{0}-n+1},
\underbrace{\ell_{0},\ell_{0}-1,\dots,\ell_{0}-n+b+2}_{n-b-1},\\
\hphantom{\xi\bigl(0,H^{c-1}_{b+m}\bigr)=|}{}
\underbrace{\ell_{0}-n+b,\dots,\ell_{0}-(n+m)+1}_{b+m},\ell_{0}-(n+m)+1-c|\nonumber.
\end{gather*}
From the inequality $\ell_{0}-n+b+2\geq \ell_0-n+2$, we have
\begin{gather*}
|\underbrace{\,0,\dots,0,\stackrel{(n-b-1)\text{-th}}{\ell_{0}-n+b+1},0,\dots,0}_{n},\underbrace{\ell_{0},\ell_{0}-1,\dots,\ell_{0}-n+b+2}_{n-b-1},\\
\quad{}
\underbrace{\ell_{0}-n+b,\dots,\ell_{0}-(n+m)+1}_{b+m},\ell_{0}-(n+m)+1-c|\nonumber.
\end{gather*}
Therefore,
\[
X={(-1)}^{b}{(-1)}^{n-b-1}\xi\bigl(0,H^{c-1}_{b+m}\bigr)={(-1)}^{n-1}\xi\bigl(0,H^{c-1}_{b+m}\bigr).
\]
 Thus, $\xi\smash{\bigl(H^{a-1}_{b},H^{c-1}_{d}\bigr)}={(-1)}^{a}\xi\smash{\bigl(0,H^{c-1}_{b+m}\bigr)}$.
\end{proof}

\begin{Proposition}\label{abcd}
Suppose that $m\geq 1$, $a, c\geq 1$ and $a+d\ne m$.
\begin{enumerate}
\item[$(1)$] For $a\geq m+1$, $d\geq m$, we have
\begin{gather}
\sum_{i=0}^{n+m}{(-1)}^{i}\bigl(S_{k_0 k_1\dots k_{n-2}\ell_i/{H^{a-1}_{b}}}(u)\cdot S_{\ell_0 \ell_1\dots \widehat{\ell_{i}}\dots\ell_{n+m}/{H^{c-1}_{d}}}(u)\nonumber\\
\hphantom{\sum_{i=0}^{n+m}{(-1)}^{i}\bigl(}{}
+S_{k_0 k_1\dots k_{n-2}\ell_i/{H^{c-1+m}_{b}}}(u)\cdot S_{\ell_0 \ell_1\dots \widehat{\ell_{i}}\dots\ell_{n+m}/{H^{a-1-m}_{d}}}(u)\nonumber\\
\hphantom{\sum_{i=0}^{n+m}{(-1)}^{i}\bigl(}{}
+S_{k_0 k_1\dots k_{n-2}\ell_i/{H^{a-1}_{d-m}}}(u)\cdot S_{\ell_0 \ell_1\dots \widehat{\ell_{i}}\dots\ell_{n+m}/{H^{c-1}_{b+m}}}(u)\nonumber\\
\hphantom{\sum_{i=0}^{n+m}{(-1)}^{i}\bigl(}{}
+S_{k_0 k_1\dots k_{n-2}\ell_i/{H^{c-1+m}_{d-m}}}(u)\cdot S_{\ell_0 \ell_1\dots \widehat{\ell_{i}}\dots\ell_{n+m}/{H^{a-1-m}_{b+m}}}(u)\bigr)=0.\label{abcd1}
\end{gather}

\item[$(2)$] For $a\geq m+1$, $0\leq d\leq m-1$, we have
 \begin{gather}
\sum_{i=0}^{n+m}{(-1)}^{i}\bigl(S_{k_0 k_1\dots k_{n-2}\ell_i/{H^{a-1}_{b}}}(u)\cdot S_{\ell_0 \ell_1\dots \widehat{\ell_{i}}\dots\ell_{n+m}/{H^{c-1}_{d}}}(u)\nonumber\\
\hphantom{\sum_{i=0}^{n+m}{(-1)}^{i}\bigl(}{}
+S_{k_0 k_1\dots k_{n-2}\ell_i/{H^{c-1+m}_{b}}}(u)\cdot S_{\ell_0 \ell_1\dots \widehat{\ell_{i}}\dots\ell_{n+m}/{H^{a-1-m}_{d}}}(u)\bigr)=0.
\label{abcd2}
\end{gather}

\item[$(3)$] For $0<a\leq m,d\geq m$, we have
 \begin{gather}
\sum_{i=0}^{n+m}{(-1)}^{i}\bigl(S_{k_0 k_1\dots k_{n-2}\ell_i/{H^{a-1}_{b}}}(u)\cdot S_{\ell_0 \ell_1\dots \widehat{\ell_{i}}\dots\ell_{n+m}/{H^{c-1}_{d}}}(u)\nonumber\\
\hphantom{\sum_{i=0}^{n+m}{(-1)}^{i}\bigl(}{}
+S_{k_0 k_1\dots k_{n-2}\ell_i/{H^{a-1}_{d-m}}}(u)\cdot S_{\ell_0 \ell_1\dots \widehat{\ell_{i}}\dots\ell_{n+m}/{H^{c-1}_{b+m}}}(u)\bigr)=0.
\label{abcd3}
\end{gather}

\item[$(4)$] For $0<a\leq m$, $0\leq d\leq m-1$, we have
\begin{gather}
\sum_{i=0}^{n+m}{(-1)}^{i}S_{k_0 k_1\dots k_{n-2}\ell_i/{H^{a-1}_{b}}}(u)\cdot S_{\ell_0 \ell_1\dots \widehat{\ell_{i}}\dots\ell_{n+m}/{H^{c-1}_{d}}}(u)=0.
\label{abcd4}
\end{gather}
\end{enumerate}
\end{Proposition}
\begin{proof} (1) We consider separately the cases where (i) $d<b+m$, (ii) $d>b+m$ and (iii)~${d=b\!+\!m}$.

Case (i) $d<b+m$.
\begin{itemize}
\item \smash{$\xi\bigl(H^{a-1}_{b}, H^{c-1}_{d}\bigr)$}: From the conditions $a\geq m+1$, $d\geq m$ and $d<b+m$, we have $\ell_{0}-n+1-a\leq \ell_{0}-(n+m)$, $\ell_{0}-(n+m)+d+2\geq \ell_{0}-n+2$ and $\ell_{0}-(n+m)+d+2 <\ell_{0}-n+b+2$, respectively. Hence, we have
\begin{gather*}
\xi\bigl(H^{a-1}_{b}, H^{c-1}_{d}\bigr)=|\underbrace{\ell_{0},\dots,\ell_{0}-n+b+2}_{n-b-1},\ell_{0}-n+b,\dots, \stackrel{(n+m-d-2)\text{-th}}{\ell_{0}-(n+m)+d+1},\nonumber\\
\hphantom{\xi\bigl(H^{a-1}_{b}, H^{c-1}_{d}\bigr)=|}{}
\underbrace{\ell_{0}-(n+m)+d,\dots,\ell_0-n+1}_{d-m},\ell_{0}-n+1-a,\nonumber\\
\hphantom{\xi\bigl(H^{a-1}_{b}, H^{c-1}_{d}\bigr)=|}{}
\underbrace{\ell_{0},\ell_{0}-1,\dots,\ell_{0}-(n+m)+(d+2)}_{n+m-d-1},\nonumber\\
\hphantom{\xi\bigl(H^{a-1}_{b}, H^{c-1}_{d}\bigr)=|}{}
\underbrace{\ell_{0}-(n+m)+d,\dots,\ell_{0}-(n+m)+1}_{d},\ell_{0}-(n+m)+1-c|.
\end{gather*}
By performing column operations, we have
\begin{gather*}
\xi\bigl(H^{a-1}_{b}, H^{c-1}_{d}\bigr)=|\underbrace{0,\dots,0}_{n+m-d-2},\stackrel{(n+m-d-2)\text{-th}}{\ell_{0}-(n+m)+d+1},0, \dots,0,\ell_{0}-n+1-a,\nonumber\\
\hphantom{\xi\bigl(H^{a-1}_{b}, H^{c-1}_{d}\bigr)=|}{}
\underbrace{\ell_{0},\ell_{0}-1,\dots,\ell_{0}-(n+m)+(d+2)}_{n+m-d-1},\nonumber\\
\hphantom{\xi\bigl(H^{a-1}_{b}, H^{c-1}_{d}\bigr)=|}{}
\underbrace{\ell_{0}-(n+m)+d,\dots,\ell_{0}-(n+m)+1}_{d},\ell_{0}-(n+m)+1-c|.
\end{gather*}
Using the Laplace expansion, we have
\begin{align*}
\xi\bigl(H^{a-1}_{b}, H^{c-1}_{d}\bigr)={}&
{(-1)}^{-m+d+1}\xi\bigl(H^{a-1}_{b}, H^{c-1}_{d}\bigr)^{0,1,\dots,n-2}_{0,1,\dots,\widehat{n+m-d-2},\dots,n-1}\\
&\quad{}\times\xi\bigl(H^{a-1}_{b}, H^{c-1}_{d}\bigr)^{n-1,n,\dots,2n+m-1}_{n+m+d-2,n,\dots,2n+m-1}\\
&{}+\xi\bigl(H^{a-1}_{b}, H^{c-1}_{d}\bigr)^{0,1,\dots,n-2}_{0,1,\dots,n-2}\cdot\xi\bigl(H^{a-1}_{b}, H^{c-1}_{d}\bigr)^{n-1,n,\dots,2n+m-1}_{n-1,n,\dots,2n+m-1}.
\end{align*}

\item $\xi\bigl(H^{c-1+m}_{b}, H^{a-1-m}_{d}\bigr)$: From (\ref{ind1}), we have
\begin{gather*}
\xi\bigl(H^{c-1+m}_{b}, H^{a-1-m}_{d}\bigr)=|\underbrace{\ell_{0},\ell_{0}-1,\dots,\ell_{0}-n+b+2}_{n-b-1},\\
\hphantom{\xi\bigl(H^{c-1+m}_{b}, H^{a-1-m}_{d}\bigr)=|}{}
\underbrace{\ell_{0}-n+b,\dots,\ell_{0}-(n+m)+d+1,\dots,\ell_{0}-n+1}_{b},\nonumber\\
\hphantom{\xi\bigl(H^{c-1+m}_{b}, H^{a-1-m}_{d}\bigr)=|}{}
\,\ell_{0}-(n+m)+1-c,\underbrace{\ell_{0},\ell_{0}-1,\dots,\ell_{0}-(n+m)+d+2}_{n+m-d-1},\nonumber\\
\hphantom{\xi\bigl(H^{c-1+m}_{b}, H^{a-1-m}_{d}\bigr)=|}{}
\underbrace{\ell_{0}-(n+m)+d,\dots,\ell_{0}-(n+m)+1}_{d},\ell_{0}-n+1-a|.
\end{gather*}
By column operations, we have
\begin{gather*}
\xi\bigl(H^{c-1+m}_{b}, H^{a-1-m}_{d}\bigr)=|\underbrace{0,\dots,0}_{n+m-d-2},\stackrel{(n+m-d-2)\text{-th}}{\ell_{0}-(n+m)+d+1},\underbrace{0,\dots,0}_{d-m},\nonumber\\
\hphantom{\xi\bigl(H^{c-1+m}_{b}, H^{a-1-m}_{d}\bigr)=|}{}
\,\ell_{0}-(n+m)+1-c,\underbrace{\ell_{0},\ell_{0}-1,\dots,\ell_{0}-(n+m)+(d+2)}_{n+m-d-1},\nonumber\\
\hphantom{\xi\bigl(H^{c-1+m}_{b}, H^{a-1-m}_{d}\bigr)=|}{}
\underbrace{\ell_{0}-(n+m)+d,\dots,\ell_{0}-(n+m)+1}_{d},\ell_{0}-n+1-a|.
\end{gather*}
Using the Laplace expansion, we have
\begin{gather*}
\xi\bigl(H^{c-1+m}_{b}, H^{a-1-m}_{d}\bigr)\\
\qquad{}={(-1)}^{-m+d+1}\xi\bigl(H^{c-1+m}_{b}, H^{a-1-m}_{d}\bigr)^{0,1,\dots,n-2}_{0,1,\dots,\widehat{n+m-d-2},\dots,n-1}\\
\qquad\qquad{}\times
\xi\bigl(H^{c-1+m}_{b}, H^{a-1-m}_{d}\bigr)^{n-1,n,\dots,2n+m-1}_{n+m+d-2,n,\dots,2n+m-1}\\
\qquad\quad{}
+\xi\bigl(H^{c-1+m}_{b}, H^{a-1-m}_{d}\bigr)^{0,1,\dots,n-2}_{0,1,\dots,n-2}\cdot\xi\bigl(H^{c-1+m}_{b}, H^{a-1-m}_{d}\bigr)^{n-1,n,\dots,2n+m-1}_{n-1,n,\dots,2n+m-1}.
\end{gather*}

\item $\xi\bigl(H^{a-1}_{d-m}, H^{c-1}_{b+m}\bigr)$: From (\ref{ind2}), we have
\begin{gather*}
\xi\bigl(H^{a-1}_{d-m}, H^{c-1}_{b+m}\bigr)=|\underbrace{\ell_{0},\dots,\ell_{0}-n+b+1,\dots,\ell_0-(n+m)+d+2}_{n+m-d-1},\nonumber\\
\hphantom{\xi\bigl(H^{a-1}_{d-m}, H^{c-1}_{b+m}\bigr)=|}{}
\underbrace{\ell_{0}-(n+m)+d,\dots,\ell_{0}-n+1}_{-m+d},\ell_{0}-n+1-a,\nonumber\\
\hphantom{\xi\bigl(H^{a-1}_{d-m}, H^{c-1}_{b+m}\bigr)=|}{}
\underbrace{\ell_{0},\ell_{0}-1,\dots,\ell_{0}-n+b+2}_{n-b-1},\nonumber\\
\hphantom{\xi\bigl(H^{a-1}_{d-m}, H^{c-1}_{b+m}\bigr)=|}{}
\underbrace{\ell_{0}-n+b,\dots,\ell_{0}-(n+m)+1}_{m+b},\ell_{0}-(n+m)+1-c|.
\end{gather*}
From the condition $d<b+m$, we have $\ell_0-n+b+2>\ell_0-n-m+d+2>\ell_0-n-m+1$. Also, we have $\ell_0-n-m+d<\ell_0-n+b$. From (\ref{ind2}), we have
\begin{gather*}
\begin{split}
& \xi\bigl(H^{a-1}_{d-m}, H^{c-1}_{b+m}\bigr)=|\underbrace{0,\dots,0}_{n-b-1},\stackrel{(n-b-1)\text{-th}}{\ell_{0}-n+b+1}, \underbrace{0,\dots,0}_{b-1},\ell_{0}-n+1-a,\\
& \hphantom{\xi\bigl(H^{a-1}_{d-m}, H^{c-1}_{b+m}\bigr)=|}{}
\underbrace{\ell_{0},\ell_{0}-1,\dots,\ell_{0}-n+b+2}_{n-b-1},\\
& \hphantom{\xi\bigl(H^{a-1}_{d-m}, H^{c-1}_{b+m}\bigr)=|}{}
\underbrace{\ell_{0}-n+b,\dots,\ell_{0}-(n+m)+1}_{m+b},\ell_{0}-(n+m)+1-c|.
\end{split}
\end{gather*}
Using the Laplace expansion, we have
\begin{align*}
\xi\bigl(H^{a-1}_{d-m}, H^{c-1}_{b+m}\bigr)={}&
{(-1)}^{b}\xi\bigl(H^{a-1}_{d-m}, H^{c-1}_{b+m}\bigr)^{0,1,\dots,n-2}_{0,1,\dots,\widehat{n-b-1},\dots,n-1}\\
&\quad{}\times
 \xi\bigl(H^{a-1}_{d-m}, H^{c-1}_{b+m}\bigr)^{n-1,n,\dots,2n+m-1}_{n-b-1,n,\dots,2n+m-1}\\
&{}+\xi\bigl(H^{a-1}_{d-m}, H^{c-1}_{b+m}\bigr)^{0,1,\dots,n-2}_{0,1,\dots,n-2}\cdot\xi\bigl(H^{a-1}_{d-m}, H^{c-1}_{b+m}\bigr)^{n-1,n,\dots,2n+m-1}_{n-1,n,\dots,2n+m-1}.
\end{align*}

\item $\xi\bigl(H^{c-1+m}_{d-m}, H^{a-1-m}_{b+m}\bigr)$: From (\ref{ind3}), we have
\begin{gather*}
\xi\bigl(H^{c-1+m}_{d-m}, H^{a-1-m}_{b+m}\bigr)=|\underbrace{\ell_{0},\ell_{0}-1,\dots,\ell_0-n+b+1,\dots,\ell_{0}-(n+m)+d+2}_{n+m-d-1},\\
\hphantom{\xi\bigl(H^{c-1+m}_{d-m}, H^{a-1-m}_{b+m}\bigr)=|}{}
\underbrace{\ell_{0}-(n+m)+d,\dots,\ell_{0}-n+1}_{-m+d},\ell_{0}-(n+m)+1-c,\nonumber\\
\hphantom{\xi\bigl(H^{c-1+m}_{d-m}, H^{a-1-m}_{b+m}\bigr)=|}{}
\underbrace{\ell_{0},\ell_{0}-1,\dots,\ell_{0}-n+b+2}_{n-b-1},\nonumber\\
\hphantom{\xi\bigl(H^{c-1+m}_{d-m}, H^{a-1-m}_{b+m}\bigr)=|}{}
\underbrace{\ell_{0}-n+b,\dots,\ell_{0}-(n+m)+1}_{m+b},\ell_{0}-n+1-a|.
\end{gather*}
By column operations, we have
\begin{gather*}
\xi\bigl(H^{c-1+m}_{d-m}, H^{a-1-m}_{b+m}\bigr)=|\underbrace{0,\dots,0}_{n-b-1},\stackrel{(n-b-1)\text{-th}}{\ell_{0}-n+b+1},0, \dots,0,\ell_{0}-(n+m)+1-c,\nonumber\\
\hphantom{\xi\bigl(H^{c-1+m}_{d-m}, H^{a-1-m}_{b+m}\bigr)=|}{}
\underbrace{\ell_{0},\ell_{0}-1,\dots,\ell_{0}-n+b+2}_{n-b-1},\nonumber\\
\hphantom{\xi\bigl(H^{c-1+m}_{d-m}, H^{a-1-m}_{b+m}\bigr)=|}{}
\underbrace{\ell_{0}-n+b,\dots,\ell_{0}-(n+m)+1}_{m+b},\ell_{0}-(n+m)+1-a|\nonumber.
\end{gather*}
Using the Laplace expansion, we have
\begin{gather*}
\xi\bigl(H^{c-1+m}_{d-m}, H^{a-1-m}_{b+m}\bigr)\\
\qquad{}={(-1)}^{b}\xi\bigl(H^{c-1+m}_{d-m}, H^{a-1+m}_{b+m}\bigr)^{0,1,\dots,n-2}_{0,1,\dots,\widehat{n-b-1},\dots,n-1}\\
\qquad\qquad{}\times \xi\bigl(H^{c-1+m}_{d-m}, H^{a-1-m}_{b+m}\bigr)^{n-1,n,\dots,2n+m-1}_{n-b-1,n,\dots,2n+m-1}\\
\qquad\quad{}+\xi\bigl(H^{c-1+m}_{d-m}, H^{a-1+m}_{b+m}\bigr)^{0,1,\dots,n-2}_{0,1,\dots,n-2}\cdot\xi\bigl(H^{c-1+m}_{d-m}, H^{a-1+m}_{b+m}\bigr)^{n-1,n,\dots,2n+m-1}_{n-1,n,\dots,2n+m-1}.
\end{gather*}
It is verified that the following holds:
\begin{gather*}
{(-1)}^{-m+d+1}\xi\bigl(H^{a-1}_{b}, H^{c-1}_{d}\bigr)^{0,1,\dots,n-2}_{0,1,\dots,\widehat{n+m-d-2},\dots,n-1}\cdot\xi\bigl(H^{a-1}_{b}, H^{c-1}_{d}\bigr)^{n-1,n,\dots,2n+m-1}_{n+m+d-2,n,\dots,2n+m-1}\\
\qquad{}=-{(-1)}^{b}\xi\bigl(H^{a-1}_{d-m}, H^{c-1}_{b+m}\bigr)^{0,1,\dots,n-2}_{0,1,\dots,\widehat{n-b-1},\dots,n-1}\cdot \xi\bigl(H^{a-1}_{d-m}, H^{c-1}_{b+m}\bigr)^{n-1,n,\dots,2n+m-1}_{n-b-1,n,\dots,2n+m-1},\\
\xi\bigl(H^{a-1}_{b}, H^{c-1}_{d}\bigr)^{0,1,\dots,n-2}_{0,1,\dots,n-2}\cdot\xi\bigl(H^{a-1}_{b}, H^{c-1}_{d}\bigr)^{n-1,n,\dots,2n+m-1}_{n-1,n,\dots,2n+m-1}\\
\qquad{}=-\xi\bigl(H^{c-1+m}_{b}, H^{a-1-m}_{d}\bigr)^{0,1,\dots,n-2}_{0,1,\dots,n-2}\cdot\xi\bigl(H^{c-1+m}_{b}, H^{a-1-m}_{d}\bigr)^{n-1,n,\dots,2n+m-1}_{n-1,n,\dots,2n+m-1},\\
{(-1)}^{-m+d+1}\xi\bigl(H^{c-1+m}_{b}, H^{a-1-m}_{d}\bigr)^{0,1,\dots,n-2}_{0,1,\dots,\widehat{n+m-d-2},\dots,n-1}\\
\quad{}\times \xi\bigl(H^{c-1+m}_{b}, H^{a-1-m}_{d}\bigr)^{n-1,n,\dots,2n+m-1}_{n+m+d-2,n,\dots,2n+m-1}\\
\qquad{}=-{(-1)}^{b}\xi\bigl(H^{c-1+m}_{d-m}, H^{a-1+m}_{b+m}\bigr)^{0,1,\dots,n-2}_{0,1,\dots,\widehat{n-b-1},\dots,n-1}\\
\qquad\quad{}\times \xi\bigl(H^{c-1+m}_{d-m}, H^{a-1-m}_{b+m}\bigr)^{n-1,n,\dots,2n+m-1}_{n-b-1,n,\dots,2n+m-1},\\
\xi\bigl(H^{a-1}_{d-m}, H^{c-1}_{b+m}\bigr)^{0,1,\dots,n-2}_{0,1,\dots,n-2}\cdot\xi\bigl(H^{a-1}_{d-m}, H^{c-1}_{b+m}\bigr)^{n-1,n,\dots,2n+m-1}_{n-1,n,\dots,2n+m-1}\\
\qquad{}=-\xi\bigl(H^{c-1+m}_{d-m}, H^{a-1+m}_{b+m}\bigr)^{0,1,\dots,n-2}_{0,1,\dots,n-2}\cdot\xi\bigl(H^{c-1+m}_{d-m}, H^{a-1+m}_{b+m}\bigr)^{n-1,n,\dots,2n+m-1}_{n-1,n,\dots,2n+m-1}.
\end{gather*}
\end{itemize}
Thus the assertion holds.

Case (ii) $d>b+m$.
Let $d<b+m$. Then we have $b+m>d-m+m$. Hence, if we replace \smash{$\xi\bigl(H^{a-1}_{b}, H^{c-1}_{d}\bigr)$} by \smash{$\xi\bigl(H^{a-1}_{d-m}, H^{c-1}_{b+m}\bigr)$} in Case (i), we obtain the equation (\ref{abcd1}).

Case (iii) $d=b+m$.
\begin{gather*}
\hspace{-3.85mm}\bullet\hspace{2mm}
\xi\bigl(H^{a-1}_{b}, H^{c-1}_{d}\bigr)=\xi\bigl(H^{a-1}_{d-m}, H^{c-1}_{b+m}\bigr)\\
\hphantom{\hspace{-3.85mm}\bullet\hspace{2mm}\xi\bigl(H^{a-1}_{b}, H^{c-1}_{d}\bigr)}{}
=|\underbrace{0,\dots,0}_{n-1},\ell_{0}-n+1-a,\underbrace{\ell_{0},\ell_{0}-1,\dots,\ell_{0}-n+b+2}_{n-b-1},\\
\hphantom{\hspace{-3.85mm}\bullet\hspace{2mm}\xi\bigl(H^{a-1}_{b}, H^{c-1}_{d}\bigr)=|}{}
\underbrace{\ell_{0}-n+b,\dots,\ell_{0}-(n+m)+1}_{b+m},\ell_{0}-(n+m)+1-c|.\\
\hspace{-3.85mm}\bullet\hspace{2mm}
\xi\bigl(H^{c-1+m}_{b}, H^{a-1-m}_{d}\bigr)=\xi\bigl(H^{c-1+m}_{d-m},H^{a-1-m}_{b+m}\bigr)\\
\hphantom{\hspace{-3.85mm}\bullet\hspace{2mm}\xi\bigl(H^{c-1+m}_{b}, H^{a-1-m}_{d}\bigr)}{}
=|\underbrace{0,\dots,0}_{n-1},\ell_{0}-n-m+1-c,\underbrace{\ell_{0},\ell_{0}-1,\dots,\ell_{0}-n+b+2}_{n-b-1},\\
\hphantom{\hspace{-3.85mm}\bullet\hspace{2mm}\xi\bigl(H^{c-1+m}_{b}, H^{a-1-m}_{d}\bigr)=|}{}
\underbrace{\ell_{0}-n+b,\dots,\ell_{0}-(n+m)+1}_{b+m},\ell_{0}-n+1-a|.
\end{gather*}
Clearly, we have \smash{$\xi\bigl(H^{a-1}_{b}, H^{c-1}_{d}\bigr)=-\xi\bigl(H^{c-1+m}_{b}, H^{a-1-m}_{d}\bigr)$}.

(2) From the condition $0\leq d\leq m-1$, we have $\ell_{0}-(n+m)+d+2\leq \ell_{0}-n+1$. By~performing column operations, we have
\begin{gather*}
\xi\bigl(H^{a-1}_{b},H^{c-1}_{d}\bigr)=|\underbrace{\ell_0,\dots,0}_{n-1},\ell_{0}-(n-1)-a,\\
\hphantom{\xi\bigl(H^{a-1}_{b},H^{c-1}_{d}\bigr)=|}{}
\underbrace{\ell_{0},\ell_{0}-1,\dots,\ell_{0}-(n+m)+(d+2)}_{n+m-d-1},\nonumber\\
\hphantom{\xi\bigl(H^{a-1}_{b},H^{c-1}_{d}\bigr)=|}{}
\underbrace{\ell_{0}-(n+m)+d,\dots,\ell_{0}-(n+m)+1}_{d},\ell_{0}-(n+m)+1-c|\nonumber.
\end{gather*}
Similarly,
\begin{gather*}
\xi\bigl(H^{c-1+m}_{b},H^{a-1-m}_{d}\bigr)=|\underbrace{0,\dots,0}_{n-1},\ell_{0}-(n+m-1)-c,\\
\hphantom{\xi\bigl(H^{c-1+m}_{b},H^{a-1-m}_{d}\bigr)=|}{}
\underbrace{\ell_{0},\ell_{0}-1,\dots,\ell_{0}-(n+m)+(d+2)}_{n+m-d-1},\\
\hphantom{\xi\bigl(H^{c-1+m}_{b},H^{a-1-m}_{d}\bigr)=|}{}
\underbrace{\ell_{0}-(n+m)+d,\dots,\ell_{0}-(n+m)+1}_{d},\ell_{0}-(n-1)-a|.
\end{gather*}
Clearly, we have \smash{$\xi\bigl(H^{a-1}_{b},H^{c-1}_{d}\bigr)=-\xi\bigl(H^{c-1+m}_{b},H^{a-1-m}_{d}\bigr)$}.

(3) We consider separately the cases where $({\rm i})\ d<b+m,\,({\rm ii})\ d>b+m$ and $({\rm iii})\ d=b+m$.

Case (i) $d<b+m$.
\begin{itemize}
\item \smash{$\xi\bigl(H^{a-1}_{b},H^{c-1}_{d}\bigr)$}: $({\rm i})$ From the conditions $d< b+m$ and $d\geq m$, we have $\ell_{0}-n+b+2>\ell_{0}-(n+m)+d+2$ and $\ell_{0}-(n+m)+d+2\geq \ell_{0}-n+2$, respectively. 
From $(\ref{ind})$, we~have
\begin{gather*}
\xi\bigl(H^{a-1}_{b},H^{c-1}_{d}\bigr)=|\underbrace{0,\dots,0}_{n-b-1},\underbrace{0,\dots,0,\ell_0-(n+m)+d+1,\dots,\ell_{0}-n+1}_{b},\nonumber\\
\hphantom{\xi\bigl(H^{a-1}_{b},H^{c-1}_{d}\bigr)=|}{}
\,\ell_{0}-(n-1)-a,\underbrace{\ell_{0},\ell_{0}-1,\dots,\ell_{0}-(n+m)+(d+2)}_{n+m-d-1},\nonumber\\
\hphantom{\xi\bigl(H^{a-1}_{b},H^{c-1}_{d}\bigr)=|}{}
\underbrace{\ell_{0}-(n+m)+d,\dots,\ell_{0}-(n+m)+1}_{d},\ell_{0}-(n+m)+1-c|.
\end{gather*}
Furthermore, from the condition $1\leq a\leq m\leq d$, we have $\ell_{0}-(n-1)-a\geq \ell_{0}-(n+m)+1$. Hence, we have
\begin{gather*}
\xi\bigl(H^{a-1}_{b},H^{c-1}_{d}\bigr)=|\underbrace{0,\dots,0}_{n-b-1},
0,\dots,0,\stackrel{(n+m-d-2)\text{-th}}{\ell_0-(n+m)+d+1},\underbrace{0,\dots,0,0}_{-m+d+2},\nonumber\\
\hphantom{\xi\bigl(H^{a-1}_{b},H^{c-1}_{d}\bigr)=|}{}
\underbrace{\ell_{0},\ell_{0}-1,\dots,\ell_{0}-(n+m)+(d+2)}_{n+m-d-1},\nonumber\\
\hphantom{\xi\bigl(H^{a-1}_{b},H^{c-1}_{d}\bigr)=|}{}
\underbrace{\ell_{0}-(n+m)+d,\dots,\ell_{0}-(n+m)+1}_{d},\ell_{0}-(n+m)+1-c|.
\end{gather*}
Let $X$ denote the above determinant with rearranged columns,
\begin{gather*}
| \underbrace{0,\dots,0}_{n-1},\ell_{0},\ell_{0}-1,\dots,
\dots,\ell_{0}-(n+m)+1,\ell_{0}-(n+m)+1-c|.
\end{gather*}
Clearly, we have \smash{$X={(-1)}^{n+1}\xi\bigl(H^{a-1}_{b},H^{c-1}_{d}\bigr)$}.

\item \smash{$\xi\bigl(H^{a-1}_{d-m},H^{c-1}_{b+m}\bigr)$}: From the condition $d< b+m$, we have $b+m>(d-m)+m$. Hence, we~have
\begin{gather*}
\xi\bigl(H^{a-1}_{d-m},H^{c-1}_{b+m}\bigr)=|\underbrace{0,\dots,0,\stackrel{(n-b-1)\text{-th}}{\ell_{0}-n+b+1}, 0,\dots,0}_{n-d+m-1},
\underbrace{0,\dots,0,0}_{d-m+1},\\
\hphantom{\xi\bigl(H^{a-1}_{d-m},H^{c-1}_{b+m}\bigr)=|}{}
\underbrace{\ell_{0},\ell_{0}-1,\dots,\ell_{0}-n+b+2}_{n-b-1},\\
\hphantom{\xi\bigl(H^{a-1}_{d-m},H^{c-1}_{b+m}\bigr)=|}{}
\underbrace{\ell_{0}-n+b,\dots,\ell_{0}-(n+m)+1}_{b+m},\ell_{0}-(n+m)+1-c|.
\end{gather*}
Clearly, we have \smash{$X={(-1)}^{n}\xi\bigl(H^{a-1}_{d-m},H^{c-1}_{b+m}\bigr)$}. Thus the assertion holds.
\end{itemize}

Case (ii) $d>b+m$.
We have the inequalities $\ell_0-n-m+d+2>\ell_0-n+b+2$ and $\ell_0-n+1-a\leq \ell_0-n-m+1$. Hence, we have the following equation:
\begin{gather*}
\hspace{-3.85mm}\bullet\hspace{2mm}
\xi\bigl(H^{a-1}_{b},H^{c-1}_{d}\bigr)=|\underbrace{0,\dots,0,
\stackrel{(n+m-d-1)\text{-th}}{\ell_0-n-m+d+1},0,\dots,0}_{n-b-1},\underbrace{0,\dots,0,0}_{b+1},\\
\hphantom{\hspace{-3.85mm}\bullet\hspace{2mm}\xi\bigl(H^{a-1}_{b},H^{c-1}_{d}\bigr)=|}{}
\underbrace{\ell_{0},\ell_{0}-1,\dots,\ell_{0}-n+b+2}_{n-b-1},\\
\hphantom{\hspace{-3.85mm}\bullet\hspace{2mm}\xi\bigl(H^{a-1}_{b},H^{c-1}_{d}\bigr)=|}{}
\underbrace{\ell_{0}-n+b,\dots,\ell_{0}-(n+m)+1}_{b+m},\ell_{0}-(n+m)+1-c|.
\end{gather*}
\smash{$X={(-1)}^{n}\xi\bigl(H^{a-1}_{b},H^{c-1}_{d}\bigr)$}. Similarly, we have the following equation:
\begin{gather*}
\hspace{-3.85mm}\bullet\hspace{2mm}
\xi\bigl(H^{a-1}_{d-m},H^{c-1}_{b+m}\bigr)=|\underbrace{0,\dots,0}_{n+m-d-1},
0,\dots,0,\stackrel{(n-b-2)\text{-th}}{\ell_0-n+b+1},\underbrace{0,\dots,0,0}_{b+2},\nonumber\\
\hphantom{\xi\bigl(H^{a-1}_{d-m},H^{c-1}_{b+m}\bigr)=|}{}
\underbrace{\ell_{0},\ell_{0}-1,\dots,\ell_{0}-n+b+2}_{n-b-1},\\
\hphantom{\xi\bigl(H^{a-1}_{d-m},H^{c-1}_{b+m}\bigr)=|}{}
\underbrace{\ell_{0}-n+b,\dots,\ell_{0}-(n+m)+1}_{b+m},\ell_{0}-(n+m)+1-c|.
\end{gather*}
\smash{$X={(-1)}^{n+1}\xi\bigl(H^{a-1}_{d-m},H^{c-1}_{b+m}\bigr)$}. Thus the assertion holds.

Case (iii) $d=b+m$. Clearly, we have \smash{$\xi\bigl(H^{a-1}_{d-m},H^{c-1}_{b+m}\bigr)=\xi\bigl(H^{a-1}_{b},H^{c-1}_{d}\bigr)=0$}. Thus the assertion holds.

(4) From the condition $0\leq d \leq m-1$, we have $\ell_0-(n+m)+d+2=\ell_{0}-n-m+d+2\leq \ell_{0}-n+1$. Hence, we have
\begin{gather*}
\xi\bigl(H^{a-1}_{b},H^{c-1}_{d}\bigr)=|\underbrace{0,\dots,0}_{n-b-1},
\underbrace{0,\dots, 0}_{b},\ell_{0}-n+1-a,\nonumber\\
\hphantom{\xi\bigl(H^{a-1}_{b},H^{c-1}_{d}\bigr)=|}{}
\underbrace{\ell_{0},\ell_{0}-1,\dots,\ell_{0}-(n+m)+d+2}_{n+m-d-1},\nonumber\\
\hphantom{\xi\bigl(H^{a-1}_{b},H^{c-1}_{d}\bigr)=|}{}
\underbrace{\ell_{0}-(n+m)+d,\dots,\ell_{0}-(n+m)+1}_{d},\ell_{0}-(n+m)+1-c|\nonumber.
\end{gather*}
Furthermore, from the conditions $a+d\leq m-1$ and $a+d\geq m+1$, we have $\ell_{0}-(n+m)+d+2\leq \ell_{0}-n+1-a$ and $\ell_{0}-(n+m)+d\geq \ell_{0}-n+1-a\geq \ell_{0}-(n+m)+1$.
Hence, we have
\begin{gather*}
\xi\bigl(H^{a-1}_{b},H^{c-1}_{d}\bigr)=|\underbrace{0,\dots,0}_{n-b-1},
\underbrace{0,\dots, 0}_{b},0,\nonumber\\
\hphantom{\xi\bigl(H^{a-1}_{b},H^{c-1}_{d}\bigr)=|}{}
\underbrace{\ell_{0},\ell_{0}-1,\dots,\ell_{0}-(n+m)+d+2}_{n+m-d-1},\nonumber\\
\hphantom{\xi\bigl(H^{a-1}_{b},H^{c-1}_{d}\bigr)=|}{}
\underbrace{\ell_{0}-(n+m)+d,\dots,\ell_{0}-(n+m)+1}_{d},\ell_{0}-(n+m)+1-c|=0.
\end{gather*}
Thus the assertion holds.
\end{proof}

We remark that the case $m=0$ is included in Proposition \ref{abcd}. The following differential relative Pl\"ucker relations are immediately derived from equations (\ref{abcd2})--(\ref{abcd4}).
\begin{Corollary}
For $r_1,r_2 \geq 0$, we have
\begin{gather*}
(1)\quad\sum_{i=0}^{n+m}{(-1)}^{i}S_{{r_1}+m}(\widetilde{\partial})S_{k_0 k_1\dots k_{n-2}\ell_i}(u)\cdot S_{r_1}(\widetilde{\partial})S_{\ell_0 \ell_1\dots \widehat{\ell_{i}}\dots\ell_{n+m}}(u)=0,\\
(2)\quad\sum_{i=0}^{n+m}{(-1)}^{i}S_{r_1}(-\widetilde{\partial})S_{k_0 k_1\dots k_{n-2}\ell_i}(u)\cdot S_{{r_1}+m}(-\widetilde{\partial})S_{\ell_0 \ell_1\dots \widehat{\ell_{i}}\dots\ell_{n+m}}(u)=0,\\
(3)\quad\sum_{i=0}^{n+m}{(-1)}^{i}S_{r_1}(-\widetilde{\partial})S_{k_0 k_1\dots k_{n-2}\ell_i}(u)\cdot S_{r_2}(\widetilde{\partial})S_{\ell_0 \ell_1\dots \widehat{\ell_{i}}\dots\ell_{n+m}}(u)=0.
\end{gather*}
\end{Corollary}

Suppose that $N\geq0$, $m\geq1$, $a,c\geq1$ and $a+b+c+d=N$. We set
\begin{gather*}
X_1=\bigl\{\bigl(H^{a-1}_{b},0\bigr)\mid a+b=N,\, 1\leq a \leq m\bigr\},
\\
X_2=\bigl\{\bigl(H^{a-1}_{b},0\bigr)\mid a+b=N,\, m<a\bigr\},
\\
X_3=\bigl\{\bigl(0,H^{c-1}_{d}\bigr)\mid c+d=N,\, 0\leq d \leq m-1\bigr\},
\\
X_4=\bigl\{\bigl(0,H^{c-1}_{d}\bigr)\mid c+d=N,\, m\leq d\bigr\},
\\
X_5=\bigl\{\bigl(H^{a-1}_{b},H^{c-1}_{d}\bigr)\mid a+d=m\bigr\},
\\
X_6=\bigl\{\bigl(H^{a-1}_{b},H^{c-1}_{d}\bigr)\mid a+d\ne m,\,a\geq m+1,\, d\geq m\bigr\},
\\
X_7=\bigl\{\bigl(H^{a-1}_{b},H^{c-1}_{d}\bigr)\mid a+d\ne m,\,a\geq m+1,\,0\leq d \leq m-1\bigr\},
\\
X_8=\bigl\{\bigl(H^{a-1}_{b},H^{c-1}_{d}\bigr)\mid a+d\ne m,\,1\leq a\leq m,\,d\geq m\bigr\},
\\
X_9=\bigl\{\bigl(H^{a-1}_{b},H^{c-1}_{d}\bigr)\mid a+d\ne m,\,1\leq a\leq m,\,0\leq d \leq m-1\bigr\}.
\end{gather*}
Furthermore, we set
\begin{gather*}
X=\coprod_{1\leq r\leq 9}X_{r}.
\end{gather*}
Clearly,
$X=\{(g_1,g_2)\in H(\alpha)\times H(\beta)\mid \alpha+\beta=N,\, \alpha\geq 0,\, \beta \geq 0\}$.

By Propositions \ref{mabcd} and \ref{abcd}, we have
\begin{gather*}
\sum_{(H^{a-1}_{b},H^{c-1}_{d})\in X_i}\xi\bigl(H^{a-1}_{b},H^{c-1}_{d}\bigr)=0,\qquad i=1,3,9.
\end{gather*}

Also, for the element \smash{$\bigl(H^{a-1}_{b},H^{c-1}_{d}\bigr)$} in $X_6$, the elements \smash{$\bigl(H^{c-1+m}_{b},H^{a-1-m}_{d}\bigr)$}, \smash{$\bigl(H^{a-1}_{d-m},H^{c-1}_{b+m}\bigr)$} and \smash{$\bigl(H^{c-1+m}_{d-m},H^{a-1-m}_{b+m}\bigr)$} are contained in the set $X_6$. This means that
\begin{gather*}
\sum_{(H^{a-1}_{b},H^{c-1}_{d})\in X_6}\xi\bigl(H^{a-1}_{b},H^{c-1}_{d}\bigr)=0.
\end{gather*}
Similarly, by the equations (\ref{abcd2}) and (\ref{abcd3}), we have
\begin{gather*}
\sum_{(H^{a-1}_{b},H^{c-1}_{d})\in X_i}\xi\bigl(H^{a-1}_{b},H^{c-1}_{d}\bigr)=0,\qquad i=7, 8.
\end{gather*}

We are now ready to state the main theorem of the present note. We call the formula skew Pl\"ucker relations.

\begin{Theorem}
Suppose that $m$ is even. We have
\begin{gather*}
\sum_{i=0}^{n+m}{(-1)}^{i}\sum_{(H^{a-1}_{b},H^{c-1}_{d})\in X}S_{k_0 k_1 \dots k_{n-2} \ell_i/H^{a-1}_{b}}(u)\cdot S_{\ell_0 \ell_1 \dots \widehat{\ell_i} \dots \ell_{n+m}/H^{c-1}_{d}}(u)=0.
\end{gather*}
\end{Theorem}
\begin{proof}
By Propositions \ref{00}--\ref{abcd}, we have
\begin{gather*}
\sum_{(H^{a-1}_{b},H^{c-1}_{d})\in X}\sum_{i= 0}^{n+m}{(-1)}^{i}S_{k_0 k_1 \dots k_{n-2} \ell_i/H^{a-1}_{b}}(u)\cdot S_{\ell_0 \ell_1 \dots \widehat{\ell_i} \dots \ell_{n+m}/H^{c-1}_{d}}(u)\\
\qquad{}=\sum_{r=1,3,6,7,8,9}\sum_{(H^{a-1}_{b},H^{c-1}_{d})\in X_r}\xi\bigl(H^{a-1}_{b},H^{c-1}_{d}\bigr)+\sum_{(H^{a-1}_{b},H^{c-1}_{d})\in X_2}\xi\bigl(H^{a-1}_{b},H^{c-1}_{d}\bigr)\\
\qquad\quad{}+\sum_{(H^{a-1}_{b},H^{c-1}_{d})\in X_4}\xi\bigl(H^{a-1}_{b},H^{c-1}_{d}\bigr)+\sum_{(H^{a-1}_{b},H^{c-1}_{d})\in X_5}\xi\bigl(H^{a-1}_{b},H^{c-1}_{d}\bigr)\\
\qquad{}=\sum_{r=2,4}\sum_{(H^{a-1}_{b},H^{c-1}_{d})\in X_r}\xi\bigl(H^{a-1}_{b},H^{c-1}_{d}\bigr)+\sum_{(H^{a-1}_{b},H^{c-1}_{d})\in X_5}\xi\bigl(H^{a-1}_{b},H^{c-1}_{d}\bigr).
\end{gather*}
The formulas (\ref{mabcd3}) and (\ref{mabcd4}) have opposite signs when $m$ is odd, and the same sign when $m$ is even. Therefore,
\begin{gather*}
\sum_{r=2,4}\sum_{(H^{a-1}_{b},H^{c-1}_{d})\in X_r}\xi\bigl(H^{a-1}_{b},H^{c-1}_{d}\bigr)=
\begin{cases}
0 & \text{if $m$ is even},\\
\displaystyle 2\sum_{(H^{a-1}_{b},H^{c-1}_{d})\in X_2}\xi\bigl(H^{a-1}_{b},H^{c-1}_{d}\bigr) & \text{otherwise}.
\end{cases}
\end{gather*}
Next, for \smash{$\bigl(H^{a-1}_{b},H^{c-1}_{d}\bigr) \in X_5$}, we have
\begin{gather*}
\#\bigl\{\bigl(H^{a-1}_{b},H^{c-1}_{d}\bigr)\bigr\}+\#\bigl\{\bigl(H^{a-1-j}_{b},H^{c-1}_{d+j}\bigr)\mid 1\leq j\leq a-1 \bigr\}\\
\qquad{}+\#\bigl\{\bigl(H^{a-1+r}_{b},H^{c-1}_{d-r}\bigr)\mid 1\leq r\leq d \bigr\}=1+(a-1)+d=m.
\end{gather*}
Hence, for \smash{$\bigl(H^{a-1}_{b},H^{c-1}_{d}\bigr) \in X_5$}, we have
\begin{gather*}
\xi\bigl(H^{a-1}_{b},H^{c-1}_{d}\bigr)+\sum_{1\leq j \leq a-1}\xi\bigl(H^{a-1-j}_{b},H^{c-1}_{d+j}\bigr)+\sum_{1\leq r \leq d}\xi\bigl(H^{a-1+r}_{b},H^{c-1}_{d-r}\bigr)\\
\qquad{}=\begin{cases}
0 & \text{if $m$ is even},\\
\text{nonzero} & \text{otherwise}.
\end{cases}
\end{gather*}
Hence,
\begin{gather*}
\sum_{(H^{a-1}_{b},H^{c-1}_{d})\in X_5}\xi\bigl(H^{a-1}_{b},H^{c-1}_{d}\bigr)=
\begin{cases}
0 & \text{if $m$ is even},\\
\text{nonzero} & \text{otherwise}.
\end{cases}
\end{gather*}
From the above, it follows that
\begin{align*}
\sum_{(H^{a-1}_{b},H^{c-1}_{d})\in X}\xi\bigl(H^{a-1}_{b},H^{c-1}_{d}\bigr)=0.
\tag*{\qed}
\end{align*}
\renewcommand{\qed}{}
\end{proof}
\begin{Example}

Let $n=3$, $m=2$, and $N=5$. We write the elements of $X$ as follows:
\[
\begin{tabular}{c|c|c|c|c|c}
$\bigl(H^{4}_{0},0\bigr)$ & $\bigl(H^{3}_{0},H^{0}_{0}\bigr)$ & $\bigl(H^{2}_{0},H^{1}_{0}\bigr)$ & $\bigl(H^{1}_{0},H^{2}_{0}\bigr)$ & $\bigl(H^{0}_{0},H^{3}_{0}\bigr)$ & $\bigl(0,H^{4}_{0}\bigr)$ \\[0.3mm]
$\bigl(H^{3}_{1},0\bigr)$ & $\bigl(H^{2}_{1},H^{0}_{0}\bigr)$ & $\bigl(H^{1}_{1},H^{1}_{0}\bigr)$ & $\bigl(H^{0}_{1},H^{2}_{0}\bigr)$ & $\bigl(H^{0}_{0},H^{2}_{1}\bigr)$ & $\bigl(0,H^{3}_{1}\bigr)$ \\[0.3mm]
$\bigl(H^{2}_{2},0\bigr)$ & $\bigl(H^{1}_{2},H^{0}_{0}\bigr)$ & $\bigl(H^{0}_{2},H^{1}_{0}\bigr)$ & $\bigl(H^{1}_{0},H^{1}_{1}\bigr)$ & $\bigl(H^{0}_{0},H^{1}_{2}\bigr)$ & $\bigl(0,H^{2}_{2}\bigr)$ \\[0.3mm]
$\bigl(H^{1}_{3},0\bigr)$ & $\bigl(H^{0}_{3},H^{0}_{0}\bigr)$ & $\bigl(H^{2}_{0},H^{0}_{1}\bigr)$ & $\bigl(H^{0}_{1},H^{1}_{1}\bigr)$ & $\bigl(H^{0}_{0},H^{0}_{3}\bigr)$ & $\bigl(0,H^{1}_{3}\bigr)$ \\[0.3mm]
$\bigl(H^{0}_{4},0\bigr)$ & & $\bigl(H^{1}_{1},H^{0}_{1}\bigr)$ & $\bigl(H^{1}_{0},/H^{0}_{2}\bigr)$ & & $\bigl(0,H^{0}_{4}\bigr)$ \\[0.3mm]
 & & $\bigl(H^{0}_{2},H^{0}_{1}\bigr)$ & $\bigl(H^{0}_{1},H^{0}_{2}\bigr)$ & &
\end{tabular}.
\]
Here,
\begin{gather*}
X_1=\bigl\{\bigl(H^{1}_{3},0\bigr), \bigl(H^{0}_{4},0\bigr)\bigr\},\\
X_2=\bigl\{\bigl(H^{4}_{0},0\bigr), \bigl(H^{3}_{1},0\bigr), \bigl(H^{2}_{2},0\bigr)\bigr\},\\
X_3=\bigl\{\bigl(0,H^{4}_{0}\bigr),\bigl(0,H^{3}_{1}\bigr)\bigr\},\\
X_4=\bigl\{\bigl(0,H^{2}_{2}\bigr),\bigl(0,H^{1}_{3}\bigr),\bigl(0,H^{0}_{4}\bigr)\bigr\},\\
X_5=\bigl\{\bigl(H^{1}_{0},H^{2}_{0}\bigr),\bigl(H^{0}_{0},H^{2}_{1}\bigr),\bigl(H^{1}_{1},H^{1}_{0}\bigr),\bigl(H^{0}_{1},H^{1}_{1}\bigr),\bigl(H^{1}_{2},H^{0}_{0}\bigr),\bigl(H^{0}_{2},H^{0}_{1}\bigr)\bigr\},\\
X_6=\phi, \\
X_7=\bigl\{\bigl(H^{3}_{0},H^{0}_{0}\bigr),\bigl(H^{2}_{0},H^{1}_{0}\bigr),\bigl(H^{2}_{1},H^{0}_{0}\bigr),\bigl(H^{2}_{0},H^{0}_{1}\bigr)\bigr\},\\
X_8=\bigl\{\bigl(H^{0}_{1},H^{0}_{2}\bigr),\bigl(H^{0}_{0},H^{0}_{3}\bigr),\bigl(H^{1}_{0},H^{0}_{2}\bigr),\bigl(H^{0}_{0},H^{1}_{2}\bigr)\bigr\},\\
X_9=\bigl\{\bigl(H^{0}_{3},H^{0}_{0}\bigr),\bigl(H^{0}_{2},H^{1}_{0}\bigr),\bigl(H^{1}_{1},H^{0}_{1}\bigr),\bigl(H^{0}_{1},H^{2}_{0}\bigr),\bigl(H^{1}_{0},H^{1}_{1}\bigr),\bigl(H^{0}_{0},H^{3}_{0}\bigr)\bigr\}.
\end{gather*}
From {\rm Proposition} $\ref{00}$, we have
\begin{gather*}
\xi\bigl(H^{1}_{3},0\bigr)=\xi\bigl(H^{0}_{4},0\bigr)=\xi\bigl(0,H^{4}_{0}\bigr)=\xi\bigl(0,H^{3}_{1}\bigr)=0.
\end{gather*}
Hence,
\begin{gather*}
\sum_{(H^{a-1}_{b},H^{c-1}_{d})\in X_1}\xi\bigl(H^{a-1}_{b},H^{c-1}_{d}\bigr)=\xi\bigl(H^{1}_{3},0\bigr)+\xi\bigl(H^{0}_{4},0\bigr)=0,
\\
\sum_{(H^{a-1}_{b},H^{c-1}_{d})\in X_3}\xi\bigl(H^{a-1}_{b},H^{c-1}_{d}\bigr)=\xi\bigl(0,H^{4}_{0}\bigr)+\xi\bigl(0,H^{3}_{1}\bigr)=0.
\end{gather*}
From {\rm Proposition}\ $\ref{mabcd}$, we have
\begin{gather*}
\xi\bigl(H^{1}_{0},H^{2}_{0}\bigr)=-\xi\bigl(H^{0}_{0},H^{2}_{1}\bigr)=-\xi\bigl(H^{4}_{0},0\bigr)=\xi\bigl(0,H^{2}_{2}\bigr),\\
\xi\bigl(H^{1}_{1},H^{1}_{0}\bigr)=-\xi\bigl(H^{0}_{1},H^{1}_{1}\bigr)=-\xi\bigl(H^{3}_{1},0\bigr)=\xi\bigl(0,H^{1}_{3}\bigr),\\
\xi\bigl(H^{1}_{2},H^{0}_{0}\bigr)=-\xi\bigl(H^{0}_{2},H^{0}_{1}\bigr)=-\xi\bigl(H^{2}_{2},0\bigr)=\xi\bigl(0,H^{0}_{4}\bigr).
\end{gather*}
Hence,
\begin{gather*}
\sum_{(H^{a-1}_{b},H^{c-1}_{d})\in X_2}\xi\bigl(H^{a-1}_{b},H^{c-1}_{d}\bigr)+\sum_{(H^{a-1}_{b},H^{c-1}_{d})\in X_4}\xi\bigl(H^{a-1}_{b},H^{c-1}_{d}\bigr)\\
\qquad{}=\xi\bigl(H^{4}_{0},0\bigr)+\xi\bigl(H^{3}_{1},0\bigr)+\xi\bigl(H^{2}_{2},0\bigr)+\xi\bigl(0,H^{2}_{2}\bigr)+\xi\bigl(0,H^{1}_{3}\bigr)+\xi\bigl(0,H^{0}_{4}\bigr)=0.
\end{gather*}
Also,
\begin{gather*}
\sum_{(H^{a-1}_{b},H^{c-1}_{d})\in X_5}\xi\bigl(H^{a-1}_{b},H^{c-1}_{d}\bigr)\\
\qquad{}=\xi\bigl(H^{1}_{0},H^{2}_{0}\bigr)+\xi\bigl(H^{0}_{0},H^{2}_{1}\bigr)+\xi\bigl(H^{1}_{1},H^{1}_{0}\bigr)+\xi\bigl(H^{0}_{1},H^{1}_{1}\bigr)+\xi\bigl(H^{1}_{2},H^{0}_{0}\bigr)+\xi\bigl(H^{0}_{2},H^{0}_{1}\bigr)=0.
\end{gather*}
From the equation $(\ref{abcd2})$, we have
\begin{gather*}
\xi\bigl(H^{3}_{0},H^{0}_{0}\bigr)=-\xi\bigl(H^{2}_{0},H^{1}_{0}\bigr),\qquad \xi\bigl(H^{2}_{1},H^{0}_{0}\bigr)=0,\qquad \xi\bigl(H^{2}_{0},H^{0}_{1}\bigr)=0.
\end{gather*}
Hence,
\begin{gather*}
\sum_{(H^{a-1}_{b},H^{c-1}_{d})\in X_7}\xi\bigl(H^{a-1}_{b},H^{c-1}_{d}\bigr)=\xi\bigl(H^{3}_{0},H^{0}_{0}\bigr)+\xi\bigl(H^{2}_{0},H^{1}_{0}\bigr)+\xi\bigl(H^{2}_{1},H^{0}_{0}\bigr)+\xi\bigl(H^{2}_{0},H^{0}_{1}\bigr)=0.
\end{gather*}
From the equation $(\ref{abcd3})$, we have
\begin{gather*}
\xi\bigl(H^{0}_{1},H^{0}_{2}\bigr)=-\xi\bigl(H^{0}_{0},H^{0}_{3}\bigr),\qquad \xi\bigl(H^{1}_{0},H^{0}_{2}\bigr)=0,\qquad \xi\bigl(H^{0}_{0},H^{1}_{2}\bigr)=0.
\end{gather*}
Hence,
\begin{gather*}
\sum_{(H^{a-1}_{b},H^{c-1}_{d})\in X_8}\xi\bigl(H^{a-1}_{b},H^{c-1}_{d}\bigr)=\xi\bigl(H^{0}_{1},H^{0}_{2}\bigr)+\xi\bigl(H^{0}_{0}, H^{0}_{3}\bigr)+\xi\bigl(H^{1}_{0},H^{0}_{2}\bigr)+\xi\bigl(H^{0}_{0},H^{1}_{2}\bigr)=0.
\end{gather*}
From the equation $(\ref{abcd4})$, we have
\begin{alignat*}{5}
&\xi\bigl(H^{0}_{3},H^{0}_{0}\bigr)=0,\qquad&& \xi\bigl(H^{0}_{2},H^{1}_{0}\bigr)=0,\qquad&& \xi\bigl(H^{1}_{1},H^{0}_{1}\bigr)=0,\qquad&& \xi\bigl(H^{0}_{1},H^{2}_{0}\bigr)=0,& \\
&\xi\bigl(H^{1}_{0},H^{1}_{1}\bigr)=0,\qquad&& \xi\bigl(H^{0}_{0},H^{3}_{0}\bigr)=0.&
\end{alignat*}
Hence,
\begin{gather*}
\sum_{(H^{a-1}_{b},H^{c-1}_{d})\in X_9}\xi\bigl(H^{a-1}_{b},H^{c-1}_{d}\bigr)\\
\qquad{}=\xi\bigl(H^{0}_{3},H^{0}_{0}\bigr)+\xi\bigl(H^{0}_{2},H^{1}_{0}\bigr)+\xi\bigl(H^{1}_{1},H^{0}_{1}\bigr)+\xi\bigl(H^{0}_{1},H^{2}_{0}\bigr)+\xi\bigl(H^{1}_{0},H^{1}_{1}\bigr)+\xi\bigl(H^{0}_{0},H^{3}_{0}\bigr)=0.
\end{gather*}
Therefore, we have
\begin{gather*}
\sum_{i=0}^{5}{(-1)}^{i}\sum_{(H^{a-1}_{b},H^{c-1}_{d})\in X}S_{k_0 k_1 \ell_i/H^{a-1}_{b}}(u)\cdot S_{\ell_0 \ell_1 \dots \widehat{\ell_i} \dots \ell_{5}/H^{c-1}_{d}}(u)=0.
\end{gather*}
\end{Example}

As is written in the introduction, the quadratic relations for the (skew) Schur functions proved in the paper correspond to certain addition formulas for the $\tau$-functions of the KP and modified KP hierarchies. Therefore, these may play a central roll in the theory of $\tau$-functions. As for the BKP hierarchy, it should be possible to discuss the quadratic relations satisfied by Q-functions.

\subsection*{Acknowledgements}
We are deeply grateful to the anonymous referees for their constructive comments that helped improve the paper. The funding was provided by KAKENHI (Grant No.~21K03208, 22K03260, 24K06859, 25K06947).

\pdfbookmark[1]{References}{ref}
\LastPageEnding

\end{document}